\documentclass{lmcs}
\pdfoutput=1

\usepackage[utf8]{inputenc}
\usepackage{lastpage}
\lmcsdoi{17}{4}{16}
\lmcsheading{}{\pageref{LastPage}}{}{}%
{Dec.~02,~2020}{Dec.~02,~2021}{}

\usepackage{eqparbox}
\usepackage{amsmath}
\usepackage{latexsym, amsfonts, stmaryrd, amssymb, graphicx, wrapfig, color}
\usepackage{tikz}
\usetikzlibrary{arrows,arrows.meta,shapes,snakes,automata,backgrounds,petri,
decorations.pathmorphing,positioning,fit,calc}
\usepackage[nomargin,inline,index,draft]{fixme} % Simplified management of FIXME's
\fxusetheme{color}
\usepackage{ifdraft}
\FXRegisterAuthor{michele}{anmichele}{\color{blue} {\underline{michele}}}
\usepackage{soul}

\theoremstyle{plain}
\newtheorem{definition}[thm]{Definition}
\newtheorem{theorem}[thm]{Theorem}
\newtheorem{proposition}[thm]{Proposition}
\newtheorem{example}[thm]{Example}

\newtheorem{remark}[thm]{Remark}

\def\eg{{\em e.g.}}

\usepackage{hyperref}
\usepackage{cleveref}
\usepackage[pagewise]{lineno}
\usepackage{setspace}
\newcommand{\corr}[2]{#2}
\newcommand{\corrmath}[2]{#2}
\newcommand{\aggiunta}[1]{{#1}}
\newcommand{\toset}[1]{\ensuremath{\overline{#1}}}
\newcommand{\CF}[1]{\ensuremath{\mathsf{CF}(#1)}}
\newcommand{\trans}[1]{\ensuremath{\,[\/{#1}\/\rangle}\,}
\newcommand{\pre}[1]{\ensuremath{\!~^{\bullet}{#1}}}
\newcommand{\post}[1]{\ensuremath{{#1} {^{\bullet}}}}

\newcommand{\hist}[1]{\ensuremath{\lfloor #1 \rfloor}}
\newcommand{\Pow}[1]{\ensuremath{\mathbf{2}^{#1}}}
\newcommand{\Powfin}[1]{\ensuremath{\mathbf{2}_\mathit{fin}^{#1}}}

\newcommand{\nat}{\ensuremath{\mathbb{N}}}

\newcommand{\flt}[1]{\ensuremath{[\![{#1}]\!]}}
\newcommand{\pes}{\textsc{pes}}

\newcommand{\bes}{\textsc{bes}}

\newcommand{\Ges}{\textsc{cdes}}
\newcommand{\cd}{\textsc{cd}}
\newcommand{\ctx}{\textsc{Cxt}}
\newcommand{\un}{\textsc{un}}
\newcommand{\ea}{\textsc{ea}}
\newcommand{\cn}{\textsc{on}}
\newcommand{\ca}{\textsc{cn}}
\newcommand{\gesrel}{\ensuremath{\gg}}
\newcommand{\enab}[1]{\ensuremath{[\/{#1}\/\rangle}}

\newcommand{\subnet}[3]{\ensuremath{#1|_{#2}^{#3}}}

\newcommand{\setenum}[1]{\{#1\}}
\newcommand{\setcomp}[2]{\{{#1} \mid {#2}\}}

\newcommand{\reachMark}[1]{\ensuremath{\mathcal{M}_{#1}}}

\newcommand{\firseq}[2]{\ensuremath{\mathcal{R}^{#1}_{#2}}}
\newcommand{\states}[1]{\ensuremath{\mathsf{St}(#1)}}

\newcommand{\trace}[1]{\ensuremath{\mathsf{tr}(#1)}}

\newcommand{\lead}[1]{\ensuremath{\mathit{lead}(#1)}}
\newcommand{\start}[1]{\ensuremath{\mathit{start}(#1)}}
\newcommand{\fs}{\textsc{fs}}

\newcommand{\MC}[1]{\ensuremath{\sim}}

\newcommand{\marko}[1]{\ensuremath{\mathsf{#1}}}
\newcommand{\Lab}{\ensuremath{\mathsf{L}}}

\newcommand{\pmv}[1]{\ensuremath{\mathsf{#1}}}

\newcommand{\Conf}[2]{\ensuremath{\mathsf{Conf}_{#2}(#1)}}
\newcommand{\Traces}[1]{\ensuremath{\mathsf{Tr}(#1)}}

\newcommand{\nettoes}[2]{\ensuremath{\mathcal{E}_{#1}^{#2}}}
\newcommand{\bundle}{\ensuremath{\mapsto}} %{\ensuremath{|\!\!\!\longrightarrow}}

\newcommand{\toges}[2]{\ensuremath{\mathcal{F}_{#1}({#2})}}
\newcommand{\rarc}[1]{\ensuremath{\underline{#1}}}

\newcommand{\inib}[1]{\ensuremath{\!~^{\circ}{#1}}}

\newcommand{\estonet}[2]{\ensuremath{\mathcal{N}_{#1}^{#2}}}

\newcommand{\len}[1]{\ensuremath{\mathit{size}(#1)}}

\newcommand{\tuple}[1]{\ensuremath{\langle #1\rangle}}
\newcommand{\stati}{\ensuremath{\mathsf{S}}}
\newcommand{\instate}{\ensuremath{\mathit{s}_0}}
\newcommand{\ontocn}[1]{\ensuremath{\mathcal{G}_{#1}}}
\newcommand{\transfs}[1]{\trans{#1}}
\newcommand{\lung}[1]{\ensuremath{\mathit{len}(#1)}}
\newcommand{\earel}{\ensuremath{\rightarrowtail}}
\newcommand{\ragg}[1]{\ensuremath{\mathsf{reach}(#1)}}
\newcommand{\raggname}{\ensuremath{\mathsf{reach}}}
\newcommand{\eainibset}[2]{\ensuremath{\mathcal{I}(#1,#2)}}
\newcommand{\eaallowset}[2]{\ensuremath{\mathcal{C}(#1,#2)}}
\newcommand{\eaname}{\ensuremath{\mathsf{A}}}
\newcommand{\pesname}{\ensuremath{\mathsf{P}}}
\newcommand{\besname}{\ensuremath{\mathsf{B}}}
\newcommand{\cdesname}{\ensuremath{\pmv{E}}}
\newcommand{\cdestoea}{\ensuremath{\mathcal{A}}}
\newcommand{\caname}{\ensuremath{K}}
\newcommand{\unname}{\ensuremath{U}}
\newcommand{\cnname}{\ensuremath{O}}
\newcommand{\causes}[1]{\ensuremath{\mathsf{Caus}(#1)}}
\newcommand{\Ini}[1]{\ensuremath{\mathsf{PD}(#1)}}
\newcommand{\maxbund}[1]{\ensuremath{\mathsf{MaxB}(#1)}}

\begin{document}

\title[A new operational representation of dependencies]{A new operational representation of dependencies in 
Event Structures}

\author[G.M. Pinna]{G. Michele Pinna} 
\address{Dipartimento di Matematica e Informatica, Universit\`a di Cagliari, Cagliari, Italy}  
\email{gmpinna@unica.it}

  \begin{abstract}
       The execution of an event in a complex and distributed system where the dependencies vary during the 
       evolution of the system can be represented in many ways, and one of them is to use 
       Context-Dependent Event structures. Event structures are related to Petri nets.
       The aim of this paper is to propose what can be the appropriate kind of Petri net corresponding
       to Context-Dependent Event structures, giving an 
       operational flavour to the dependencies represented in a Context/Dependent Event structure.
       Dependencies are often operationally represented, in Petri nets, by tokens produced by 
       activities and consumed by others. Here we shift the perspective using contextual arcs
       to characterize what has happened so far and in this way to describe the dependencies among
       the various activities.
  \end{abstract}
\keywords{Petri Nets, Event Structures, Operational Semantics, Contextual Nets}

\maketitle

%\linenumbers

%%%%%%%%%%%%%%%%%%%%%%%%%%%%%%%%%%%%%%%%%%%%%%%%%%%%%%%%%%%%%%%%%%%%%%%%%%%%%%%%

%%% \input{introduction}
\section{Introduction}\label{sec:intro}
Since the introduction of the notion of Event structure (\cite{NPW:PNES} and \cite{Win:ES}) 
the close relationship between this notion and suitable nets has been investigated. 
The ingredients of an event structure are, beside a set of events, a number of relations 
used to express which events can be part of a configuration (the snapshot of a concurrent 
system), modeling a consistency predicate, and how events can be added to reach another 
configuration, modeling the dependencies among the (sets of) events.
On the nets side we have transitions, modeling the activities, and places, modeling
resources the activities may need, consume or produces. These ingredients, together
with some constraints on how places and transitions are related (via flow, inhibitor or read arcs
satisfying suitable properties), can give also a more \emph{operational} description of a concurrent
and distributed system.  
Indeed the relationship between event structures and nets is grounded on the observation 
that also in (suitable) Petri nets the
relations among events are representable, as it has been done in 
\cite{GR:NSBPN83} for what concern the partial 
order and \cite{NPW:PNES} for the partial order and conflict.

Since then several notions of event structures have been proposed. 
We recall just few of them: the classical \emph{prime} event structures 
\cite{Win:ES} where
the dependency between events, called \emph{causality}, is modeled by a partial order and the consistency
is described by a symmetric \emph{conflict} relation.
Then \emph{flow} event structures \cite{Bou:FESFN} drop the requirement that the dependency should
be a partial order on the whole set of events, \emph{bundle} event structures \cite{Langerak:1992:BES}
represent OR-causality by allowing each event to be caused by a unique member of a bundle of events
(and this constraint may be relaxed).
\emph{Asymmetric} event structures \cite{BCM:CNAED}, via notion of weak causality,
model asymmetric conflicts, whereas \emph{Inhibitor} event structures \cite{BBCP:rivista} are able
to faithfully capture the dependencies among events which arise in the 
presence of read and inhibitor arcs in safe nets. 
In \cite{BCP:LenNets} a notion of event structures where the causality relation may be circular is
investigated, and in \cite{AKPN:lmcs18} the notion of dynamic causality is considered.
Finally, we mention the quite general approach presented in \cite{GP:CSESPN}, where there is
a unique relation, akin to a \emph{deduction relation}. 
To each of the mentioned event structures 
a particular class of nets is related. Prime event structures have a correspondence in 
\emph{occurrence nets}, flow event structures have flow nets whereas 
\emph{unravel nets} \cite{CaPi:PN14} are related to bundle event structures. 
Continuing we have that asymmetric and inhibitor event structures have a correspondence with  
\emph{contextual nets} \cite{BCM:CNAED,BBCP:rivista}, and event structures with circular causality
with  \emph{lending nets} \cite{BCP:LenNets}, finally to those with dynamic causality we have
\emph{inhibitor unravel nets} \cite{CP:soap17} and to the configuration structures presented in 
\cite{GP:CSESPN}  we have the notion of \emph{1-occurrence nets}.
Most of the approaches relating nets with event structures are based on the equation
``event = transition'', even if many of the events represent the same \emph{high level}
activity. 
The idea that some of the transitions may be somehow identified
as they represent the same activity is the one pursued in many works aiming at reducing
the size of the net, like \emph{merged processes} (\cite{KKKV:Acta06}), \emph{trellis processes} 
(\cite{Fab:Trellis}), \emph{merging relation}
approach  (\cite{CP:lata17}) or \emph{spread nets} (\cite{FP:Ice18} and \cite{PF:jlamp20}),
but these approaches are mostly unrelated with event structure of any kind.

\aggiunta{\cite{Pi:coord19} and \cite{Pi:lmcs} propose the notion of \emph{context-dependent} 
event structures as the event structure where (almost all) the variety of dependencies which may arise 
among events can be modeled. 
This is achieved using just 
one fairly general dependency relation, called \emph{context-dependency} relation. 
Context-dependent event 
structure are more general than many kind of event structures, as proved in 
\cite{Pi:lmcs}, with the exception of event structure with circular causality. It should however
be observed that circular dependencies may be understood as dependencies where the justification
for the happening of an event can be given at a later point, which is not the case of the kind
of event structures considered in this paper, thus we may say that context-dependent event 
structure are more general than event structures where each event may happen if a complete 
justification for its happening is present.~}
In this paper we {address} the usual problem: given {an} event 
structure, in our case a context-dependent event structures, 
{which could be the kind of} net which may \emph{correspond} to it.
To understand {the characteristics} of \aggiunta{the~}net {to} be 
{related} to {context-dependent} 
event structures
%\corr{(\cite{Pi:coord19} and \cite{Pi:lmcs})}{}
we first observe that in these event structures each 
event may happen in many different and often unrelated contexts, hence
the same event cannot have (almost) the same \emph{past} \aggiunta{event up to suitable equivalences,~}as 
it happens in many approaches, \eg{} in unravel nets, trellis processes or spread nets 
among others.
The second observation, which has been used also in \cite{Pi:coord19} and \cite{Pi:lmcs}, is that 
dependencies among transitions (events) in nets may be represented in different ways.
The usual way to represent dependencies is using the flow relation of the net, and the
dependency is then \emph{signaled} by a \emph{shared} place, but dependencies can be 
described also using \emph{contextual} arcs, \emph{i.e.} arcs testing for the presence or 
absence of tokens.  
Following these two observations we argue that each of the context that are \emph{allowing} an event
to happen can be modeled with \emph{inhibitor} (\cite{JK:SIN}) and/or \emph{read} arcs 
(\cite{MR:CN})\aggiunta{, yielding the notion of \emph{causal} nets}. 
It should be stressed
that these kind of arcs have been introduced for different purposes, but never for nets 
which are meant to describe the behaviour of another one. 

Usually places \aggiunta{in a net modeling the behaviour of a 
concurrent and distributed systems~}are seen as resources that one
(or more) activities may produce and that {are} 
consumed by anothers. 
Instead of\aggiunta{, in a causal net,~}we 
consider places as \emph{control points}: they 
\aggiunta{simply~}signal what events have been executed and thus
the marking \aggiunta{alone~}determines {the \emph{configuration} of the net, without 
the need of reconstructing which events have been executed before}.
\aggiunta{This shift in perspective, namely considering places as control points rather than resources to
be produced and consumed, gives further generality that can be used to model more composite relations.
Indeed we observe that this idea has been used in \cite{MMP:lics2021} to model the so called
\emph{out of causal order} reversibility when relating \emph{reversible} prime event structures 
\cite{iainirek} and nets, which otherwise would have been impossible, as noticed in \cite{MMPPU:rc2020}.}

We have \aggiunta{also~}to stress that, having activities quite different contexts, it is natural 
to \emph{implement} the same activity in different ways to reflect the different context.
\aggiunta{In causal nets, which are labeled nets, the same label may be associated to various transitions
that have to be considered as the different incarnations of the same activity (label).~}
This is precisely what it is done in \emph{unfoldings}, 
the main difference being that {in causal net~}we look at the event itself with 
different \emph{incarnations}
whereas in the classical approaches the event is identified with an unique incarnation. 

The approach we pursue here originated in the one we adopted for dynamic event structures
in \cite{CP:soap17}, though \aggiunta{t}here  the \emph{classical} dependencies among 
events (those called
causal dependencies) are {represented using} the {standard} 
machinery{, namely using
places as resources.~}
{There~}
we \aggiunta{also~}argued that the proper net corresponding to these kind of event 
structure\aggiunta{, which~}are meant to
give an \emph{operational} representation of what \emph{denotationally} {is} characterized 
{by} a single event, have to be represented as 
different transitions\aggiunta{ with the same label}.
The approach is a conservative one: the dependencies represented in {other} kinds on nets
can be represented also in {causal nets, possibly with some further
constraint,} and  to suitably characterized \aggiunta{causal~}nets it
is possible to associate the corresponding context-dependent event structure.
\aggiunta{Indeed, we will show that causal nets are more general than occurrence and unravel nets.}
It should be stressed that the conflicts between events in causal nets are explicitly represented
and cannot be inferred otherwise.

This paper is an extended and revised version of \cite{Pi:coord20}. We have added 
some examples and compared the notion of causal net with other \emph{semantics} net like occurrence nets
and unravel nets, and we have made more precise the relationships among those. 
The constructions associating event structures and nets turn out to follow a common pattern
once that the idea of dependency shifts from the classical \emph{produce/consume} relationship toward
the one where contextual arcs (testing positively or negatively but without consuming), giving in our opinion 
further evidence that this new operational view has a solid ground.

\subsection*{Organization of the paper.}\
In the next section we recall the notions of contextual nets, \aggiunta{on which the notion of
\emph{causal} net is based. Section~\ref{sec:netandes} reviews the more classical
approaches relating nets and event structures, namely~}occurrence net and prime event
structure and also unravel net
and bundle event structure. Furthermore the relationships
occurrence nets - prime event structures and 
unravel nets - bundle event structures are exhibited. In Section~\ref{sec:es-mie} we recall
the notion of context-dependent event structure and we discuss how a canonical representation can be
obtained. In Section~\ref{sec:causal-nets} we introduce the
notion of \emph{causal} net and we show also how occurrence nets and unravel nets can be seen as causal nets.
We also give a direct translation from prime event structures to causal net and vice versa.
In Section~\ref{sec:cdes-causal} we discuss how to associate a causal net to a 
context-dependent event structure and vice versa, showing that the notion of causal net
is adequate. Some conclusions end the paper.

\section{Preliminaries}\label{sec:preliminaries}
 We denote with  $\nat$ the set of natural numbers.
 Given a set $A$ with $\Pow{A}$ we denote the set of subsets of $A$ and
 with $\Powfin{A}$ the set of the finite subsets of $A$.
 
 Let $A$ be a set, a \emph{multiset} of $A$ is a function $m : A
 \rightarrow \nat$.
 The set of multisets of $A$ is denoted by $\mu A$.
 \corr{}{A multiset $m$ over $A$ is sometimes written as $\sum m(a)\cdot a$ and $m(a)$ is the number of
 occurrence of $a$ in $\sum m(a)\cdot a$.}  
 We assume the usual operations on multisets such as union $+$ and difference $-$.
 We write $m \subseteq m'$ if $m(a) \leq m'(a)$ for all $a \in A$.  
 For $m\in \mu A$, we denote with $\flt{m}$ the multiset defined as $\flt{m}(a)
 = 1$ if $m(a) > 0$ and $\flt{m}(a) = 0$ otherwise. 
 When a multiset $m$ of $A$ is a set, 
 \emph{i.e.} $m = \flt{m}$, we write
 $a\in m$ to denote that $m(a) \neq 0$, and often confuse the
 multiset $m$ with the set $\setcomp{a\in A}{m(a) \neq 0}$ or a subset 
 $X\subseteq A$ with the multiset $X(a) = 1$ if $a\in A$ and 
 $X(a) = 0$ otherwise.
 Furthermore we use the standard set operations like $\cap$, $\cup$ or
 $\setminus$.
 
 Given a set $A$ and a relation $<\ \subseteq A\times A$, we say that $<$ is 
 an irreflexive partial order whenever
 it is irreflexive and transitive. 
 We shall write $\leq$ for the reflexive closure of an irreflexive partial order $<$.
 Given an irreflexive relation $\prec\ \subseteq A\times A$, with $\prec^{+}$ we denote
 its transitive closure. 
 
 Given a function $f : A \to B$, $\mathit{dom}(f) = \setcomp{a\in A}{\exists b\in B.\ f(a) = b}$
 is the domain of $f$, and $\mathit{codom}(f) = \setcomp{b\in B}{\exists a\in A.\ f(a) = b}$
 is the codomain of $f$.
 
 Given a set $A$, a sequence of elements in $A$ is a partial mapping 
 $\rho : \nat\rightharpoonup A$ such that, given any  
 $n\in \nat$, if $\rho(n)$ is defined and equal to $a\in A$ then $\forall i \leq n$ 
 also $\rho(i)$ is defined.
 A sequence is finite if $|\mathit{dom}(\rho)|$ is finite, and 
 the length of a sequence $\rho$, denoted with $\len{\rho}$, 
 is the cardinality of $\mathit{dom}(\rho)$.
 A sequence $\rho$ is often written as $a_1a_2\cdots$ where $a_i = \rho(i)$.
 With $\toset{\rho}$ we denote the codomain of $\rho$.
 Requiring that a sequence $\rho$ has distinct elements accounts to stipulate that
 $\rho$ is injective on $\mathit{dom}(\rho)$.
 The sequence $\rho$ such that $\mathit{dom}(\rho) = \emptyset$, the \emph{empty} sequence,
 is denoted with $\epsilon$. 
 Finally with $\cdot$ we denote the \emph{concatenation} operator that take a finite sequence
 $\rho$ and a sequence $\rho'$ and gives the sequence $\rho''$ defined as
 $\rho''(i) = \rho(i)$ if $i\leq\len{\rho}$ and $\rho''(i) = \rho'(i-\len{\rho})$ otherwise.
 
\subsection{Contextual Petri nets} 
We briefly review the notion of (labeled) Petri net (\cite{Rei:PNI,Rei:UPN}) and its variant enriched 
with contextual arcs (\cite{MR:CN} and \cite{BBCP:rivista}) along with some auxiliary notions.

Fixed a set $\Lab$ of \emph{labels},  
we recall that a \emph{net} is the 4-tuple $N = \langle S, T, F, \mathsf{m}\rangle$
where $S$ is a set of \emph{places} (usually depicted with circles) and $T$ is a set of \emph{transitions} 
(usually depicted as squares) and $S \cap T = \emptyset$, 
$F \subseteq (S\times T)\cup (T\times S)$ is the \emph{flow} relation and 
$\mathsf{m}\in \mu S$ is called the \emph{initial marking}. A \emph{labeled} net is a 
net equipped with a \emph{labeling} mapping $\ell : T \to \Lab$. 
\begin{definition}
   A (labelled) \emph{contextual Petri net} is the tuple 
   $N = \langle S, T, F, I, R, \mathsf{m}, \ell\rangle$, where 
   \begin{itemize}
    \item $\langle S, T, F, \mathsf{m}\rangle$ is a net, 
    \item $I \subseteq S\times T$ are the \emph{inhibitor} arcs, 
    \item $R \subseteq S\times T$ are the \emph{read} arcs, and 
    \item $\ell : T \to \Lab$ is the labeling mapping and $\ell$ is a total function.
   \end{itemize}
 \end{definition}
 Inhibitor arcs depicted as lines with a circle on one end, and
 read arcs as plain lines. We sometimes omit the $\ell$ mapping when $\Lab$ is $T$ and $\ell$ is the 
 identity. We will often call a contextual Petri net as Petri net or simply net.
 In the following figure a contextual Petri net is depicted.

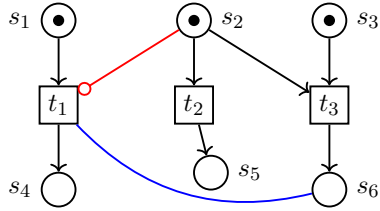
\begin{figure}[h]
\begin{center}
 \scalebox{0.9}{\begin{tikzpicture}
\tikzstyle{inhibitorred}=[o-, draw=red,thick]
\tikzstyle{pre}=[<-,thick]
\tikzstyle{post}=[->,thick]
\tikzstyle{readblue}=[-, draw=blue,thick]
\tikzstyle{transition}=[rectangle, draw=black,thick,minimum size=5mm]
\tikzstyle{place}=[circle, draw=black,thick,minimum size=5mm]
\node[place,tokens=1] (p1) at (0,2.5) [label=left:$s_1$] {};
\node[place,tokens=1] (p3) at (2,2.5) [label=right:$s_2$] {};
\node[place,tokens=1] (p5) at (4,2.5) [label=right:$s_3$] {};
\node[place] (p2) at (0,0) [label=left:$s_4$] {};
\node[place] (p4) at (2.25,0.25) [label=right:$s_5$] {};
\node[place] (p6) at (4,0) [label=right:$s_6$] {};
\node[transition] (t1) at (0,1.25)  {$t_1$}
edge[pre] (p1)
edge[inhibitorred] (p3)
edge[readblue, bend right] (p6)
edge[post] (p2);
\node[transition] (t2) at (2,1.25) {$t_2$}
edge[pre] (p3)
edge[post] (p4);
\node[transition] (t3) at (4,1.25)  {$t_3$}
edge[pre] (p5)
edge[pre] (p3)
edge[post] (p6);
\end{tikzpicture}}
\end{center}
\caption{A contextual Petri net}\label{fig:pn1}
\end{figure} 
\aggiunta{The arc connecting the place $s_6$ to the transition $t_1$ is a read arc, and the one connecting
the place $s_2$ to the same transition is an inhibitor one.}
 
 Given a net $N = \langle S, T, F, I, R, \mathsf{m}\rangle$ and  $x\in S\cup T$, we define the following
 (multi)sets:  
 $\pre{x} = \setcomp{y}{(y,x)\in F}$ and 
 $\post{x} = \setcomp{y}{(x,y)\in F}$.
 If $x\in S$ then 
 $\pre{x} \in \mu T$ and  $\post{x} \in \mu T$;
 analogously, 
 if $x\in T$ then $\pre{x}\in\mu S$ and $\post{x} \in \mu S$.
 Given a transition $t$, with $\inib{t}$ we denote the (multi)set $\setcomp{s}{(s,t)\in I}$ and
 with $\rarc{t}$ the (multi)set $\setcomp{s}{(s,t)\in R}$.

 A transitions $t\in T$
 is enabled at a marking $m\in \mu S$, denoted by $m\trans{t}$,
 whenever $\pre{t} + \rarc{t} \subseteq m$ and  $\forall s\in\flt{\inib{t}}.\ m(s) = 0$. 
 A transition $t$ enabled at a marking $m$ can \emph{fire} and its firing produces 
 the marking $m' = m - \pre{t} + \post{t}$.
 The firing of $t$ at a marking $m$ is denoted by $m\trans{t}m'$.
 $\rarc{t}$ and $\inib{t}$ are the \emph{contexts} in which the transition $t$ may fire at a marking
 $m$ provided that
 $\pre{t}\subseteq m$. $\rarc{t}$ is the \emph{positive} context as tokens must be present in
 the places connected to $t$ with a read arcs whereas $\inib{t}$ is a \emph{negative} one as 
 no token must be present in a place connected to a transition with an inhibitor arc. The tokens (or their
 absence) are just tested.
 \aggiunta{In the net in Figure~\ref{fig:pn1} the transition $t_1$ is enabled when 
 the places $s_1$ and $s_6$ are marked and the place $s_2$ is unmarked. In particular the
 token in the place $s_6$ is just tested for presence, whereas the place $s_2$ is tested for the absence
 of a token.}
 
 We assume that each transition $t$ of a net $N$ is such that $\pre{t}\neq\emptyset$,
 meaning that no transition may fire without \emph{consuming} some token, even if the 
 contexts would allow it.
 
 Given a generic marking $m$ (not necessarily  the initial one), 
 the \emph{firing sequence} ({shortened as} \fs) of  
 $N = \langle S, T, F, I, R, \mathsf{m}\rangle$  starting at $m$ is
 defined as the sequence (finite or infinite)
 $m_0\transfs{t_1}m_1\trans{t_2}m_2\cdots m_{n-1}\transfs{t_n}m_n\cdots$ 
 such that for all $i$ it holds that $m_i\trans{t_{i+1}}m_{i+1}$ and 
 where $m_0 = m$.
 The set of firing sequences of a net $N$ 
 starting at a marking $m$ is denoted by $\firseq{N}{m}$ and it is ranged over by $\sigma$.
 Given a \fs\ $\sigma = m\transfs{t_1}\sigma'\corrmath{\transfs{t_n}m_n}{}$, we denote with
 $\start{\sigma}$  the marking $m$ and\corr{with $\lead{\sigma}$ the marking $m_n$ if
 the firing sequence is finite and $\sigma = m_0\transfs{t_1}m_1\cdots m_{n-1}\transfs{t_n}m_n$}{, if 
 the firing sequence $\sigma$ is finite and $\sigma = m\transfs{t_1}m_1\cdots m_{n-1}\transfs{t_n}m_n$,
 with $\lead{\sigma}$ we denote the marking $m_n$}. 
 The \emph{length} of a \fs{}, \corr{denoted with}{written as} $\lung{\sigma}$,
 is defined as \corr{}{follows:}
 $\lung{\sigma} = 0$ if $\sigma = m$\corr{ and}{,} $\lung{\sigma} = 1 + \lung{\sigma'}$ if
 $\sigma = m\trans{t}\sigma'$ and $\sigma$ is a finite firing 
 sequence, \corr{}{and} $\lung{\sigma} = \infty$ otherwise.
 %
% $\remains{\sigma}$ denotes the 
% \fs\ $\sigma$, provided that $\sigma = m_0\trans{t_1}m_1\sigma'$.
 Given a \fs\ $\sigma = m_0\trans{t_1}m_1\trans{t_2}m_2\cdots m_{n-1}\trans{t_n}m_n\cdots$,
 with $\sigma(i)$ we denote the \fs\ $m_0\trans{t_1}m_1\trans{t_2}m_2\cdots m_{i-1}\trans{t_i}m_i$
 for all $i$ less or equal to the length of the \fs.
 Given a net $N$, a marking $m$ is \emph{reachable} 
 iff there exists a 
 \fs\ $\sigma \in \firseq{N}{\mathsf{m}}$ 
 such that $\lead{\sigma}$ is $m$. 
 The set of reachable markings of $N$ is
 $\reachMark{N} = \setcomp{\lead{\sigma}}{\sigma\in\firseq{N}{\mathsf{m}}}$. 
 \begin{definition}\label{de:states-of-a-net} 
  Let $N = \langle S, T, F, I, R, \mathsf{m}, \ell\rangle$ be a labeled contextual Petri net.
  Let $\sigma = \mathsf{m}\trans{t_1}m_1\cdots m_{n-1}\trans{t_n}m'$ be a \fs, then 
  $X_{\sigma} = \sum_{i=1}^{n} \setenum{t_i}$ is a \emph{state} of $N$.
  The set of states is \(\states{N} = \setcomp{X_{\sigma}\in \mu T}{\sigma\in\firseq{N}{\mathsf{m}}}\).
 \end{definition}
 A state is just the multiset of the transitions that have been fired in a \fs.

 \begin{definition}\label{de:conf-of-a-net} 
  Let $N = \langle S, T, F, I, R, \mathsf{m}, \ell\rangle$ be a labeled contextual Petri net.
  The set of \emph{configurations} of $N$, denoted with $\Conf{N}{}$, is 
  the set $\setcomp{\ell(X)}{X\in \states{N}}$.
 \end{definition}
 Thus a configuration is just the multiset of labels associated to a state.
% A configuration is denoted with $C$ (and a state with $X$). 
 If the net belongs to a specific kind, say ``\textsc{net}'', we may add a subscript
 \corr{to}{} $\Conf{N}{\textsc{net}}$\corr{}{\ to stress that the set of configurations refers to this
 kind of net}.
 
 Given a \fs\ $\sigma$, we can extract the sequence of the fired transitions. Formally
 $\trace{\sigma} = \epsilon$ if $\lung{\sigma} = 0$ and 
 $\trace{\sigma} = \ell(t)\cdot\trace{\sigma'}$ if $\sigma = m\trans{t}\sigma'$.
 
 \begin{definition}\label{de:traces-of-a-net} 
  Let $N = \langle S, T, F, I, R, \mathsf{m}, \ell\rangle$ be a labeled contextual Petri net,
  and $\sigma\in\firseq{N}{\mathsf{m}}$ be a \fs\ of $N$. 
  Then the set of \emph{traces} of $N$ is 
  $\Traces{N} = \setcomp{\trace{\sigma}}{\sigma\in\firseq{N}{\mathsf{m}}}$.
 \end{definition}

\begin{example}\label{ex:cont-net1}
Consider the contextual Petri net in Figure~\ref{fig:pn1}. 
At the initial marking $t_2$ and
$t_3$ are enabled whereas $t_1$ is not. After the execution of $t_2$ no other transition
is enabled. After the firing of $t_3$ the transition $t_1$ is enabled, as 
no token is present in the place $s_2$ and a token is present in the place $s_6$, the former
being connected to transition $t_1$ with an inhibitor arc and the latter being 
connected to transition $t_1$ with a read arc. 
The states of this net are $\setenum{t_2}$ and $\setenum{t_3,t_1}$.
Assume now that $\ell(t_1) = \pmv{b}$ and $\ell(t_2) = \ell(t_3) = \pmv{a}$, 
then the configurations are $\setenum{\pmv{a}}$ and $\setenum{\pmv{a},\pmv{b}}$.
Finally the traces of this net are $\pmv{a}$ and $\pmv{a}\pmv{b}$.
\end{example}

The following definitions characterize nets from a \emph{semantical} point of view.
\begin{definition}\label{de:single-execution-net}  
 A net $N = \langle S, T, F, I, R, \mathsf{m}, \ell\rangle$ is said to be \emph{safe} if
 each marking $m\in \reachMark{N}$ is such that $m = \flt{m}$. 
\end{definition}
In this paper we will consider safe nets, where each place contains at most one token.
The following definitions outline nets with respect to states, configurations and traces.

\begin{definition}\label{de:single-ex-net}
 A net $N = \langle S, T, F, I, R, \mathsf{m}, \ell\rangle$ is said to be a  
 \emph{single execution} net if
 each state  $X\in \states{N}$ is such that $X = \flt{X}$. 
\end{definition}
In a single execution net a transition $t$ in a firing sequence may be fired just once, as
the net in Example \ref{ex:cont-net1}.
In \cite{GP:CS} and \cite{GP:CSESPN} these nets (without inhibitor and read arcs) are called 
\emph{1-occurrence} net, and the name itself stress this characteristic.

\begin{definition}\label{de:unfolding-net}
 A net $N = \langle S, T, F, I, R, \mathsf{m}, \ell\rangle$ is said to be an 
 \emph{unfolding} if
 each configuration  $C\in \Conf{N}{}$ is such that $C = \flt{C}$. 
\end{definition}
Clearly each unfolding is also a single execution \corr{one}{net}, but the vice versa does not hold.
\corr{When the labeling of the net is an injective mapping we have that to each state a configuration 
corresponds and vice versa.}{Observe that if the labeling of the unfolding is an injective mapping then
states and configurations may be confused as $\Conf{N}{}$ and $\states{N}$ are bijectively related.}

\begin{remark}
 In literature \emph{unfolding} is often used to denote not only a net with suitable
 characteristic (among them the fact that each transition is fired just once in each execution),
 but also how this net is related to another one (the one to be unfolded). Here we use it
 to stress that each configuration is a set\corr{.}{~and also to point out that the nets
 we are considering are labeled ones, and labels suggest that a labeled transition is associated
 to some activity, as it happens in \emph{classical} unfoldings.} 
\end{remark}

We consider two nets \emph{equivalent} when they have the same set of transitions, the same
labeling, and the 
same set of states (which means that they have also the same configurations and the same
set of firing sequences).
\begin{definition}\label{de:netequiv}
 Let $N = \langle S, T, F, I, R, \mathsf{m}, \ell\rangle$ and 
 $N' = \langle S', T', F', I', R', \mathsf{m}', \ell'\rangle$ be two nets.
 $N$ is \emph{equivalent} to $N'$, written as $N \equiv N'$
 whenever $T = T'$, $\ell = \ell'$ and $\states{N} = \states{N'}$.
\end{definition}
Though the equivalence relation on nets implies that the nets have the same firing sequences, 
it does not imply that the sets of reachable markings is the same, hence 
$N \equiv N'$ does not imply that $\reachMark{N} = \reachMark{N'}$. 

Observe that different traces may be associated to the same \corr{state}{configuration} of the net.
In an unfolding the following propositions hold.
\begin{proposition}\label{pr:banale1}
 Let $N = \langle S, T, F, I, R, \mathsf{m}, \ell\rangle$ be an unfolding.
 Then, for each $\rho \in \Traces{N}$ it holds that $\toset{\rho}\in \Conf{N}{}$.
\end{proposition}
\begin{proof}
 In an unfolding each configuration is a set, hence the thesis.
\end{proof}

\begin{proposition}\label{pr:banale2}
 Let $N = \langle S, T, F, I, R, \mathsf{m}, \ell\rangle$ be an unfolding.
 Then, for each $X \in \Conf{N}{}$ there exists a trace $\rho \in \Traces{N}$ such that $\toset{\rho} = X$.
\end{proposition}
\begin{proof}
 As $X \in \Conf{N}{}$ then there is a \fs\ $\sigma$ such that $X_{\sigma} = X$, and take
 $\rho = \trace{\sigma}$. By Proposition~\ref{pr:banale1} it follows that $\toset{\rho} = X$.
\end{proof}

The following definition give a \emph{structural} characterization 
of when two transitions, which we say are \emph{conflicting} transitions, 
never happen together in any execution of a net.
\begin{definition}\label{de:sat-confl}
 Let $N = \langle S, T, F, I, R, \mathsf{m}, \ell\rangle$ be a net. 
 We say that $N$ is \emph{conflict saturated} if for all 
 $t, t'\in T$ such that $\forall X\in \states{N}.\ \setenum{t, t'}\not\subseteq \flt{X}$,
 it holds that $\pre{t}\cap\pre{t'} \neq \emptyset$.
\end{definition}
In a conflict saturated net the fact that two transitions never happen in a state is also
justified by the existence of a place in their preset, and this suggests that in certain
kind of nets conflicts can be characterized structurally.
Indeed each single execution net can be transformed into an equivalent one conflict saturated.
\begin{proposition}\label{pr:sat-confl}
 Let $N = \langle S, T, F, I, R, \mathsf{m}, \ell\rangle$ be a single execution net.
 Let $\mathit{Confl}(N)$ be the set $\setcomp{\setenum{t,t'}}{\forall X\in \states{N}.\ 
 \setenum{t, t'}\not\subseteq \flt{X}}$.
 Then the net $N^{\#} =
 \langle S\cup S^{\#}, T, F\cup\setcomp{(s_A,t)}{s_A\in S^{\#}\ \land\ t\in A}, I, R, \mathsf{m}\cup S^{\#}, \ell\rangle$, where $S^{\#} = \setcomp{s_A}{A\in \mathit{Confl}(N)}$, 
 is conflict saturated and 
 $N\equiv N^{\#}$.
\end{proposition}
\begin{proof}
 \corr{$N^{\#}$ is conflict saturated by construction, as it is a single execution net.
 The equivalence $N\equiv N^{\#}$ holds as the added places do not influence 
 the firing sequences of the net.}
 {It is enough to prove that $\states{N} = \states{N^{\#}}$. Observe first that the added places $S^{\#}$ in
 $N^{\#}$ do not have any incoming arcs and are initially marked. 
 This implies 
 that $\states{N^{\#}} \subseteq \states{N}$. For the vice versa, consider  
 $X\in \states{N}$. $X = X_{\sigma}$ for some
 \fs{} $\sigma\in\firseq{N}{\mathsf{m}}$. We construct, by induction on the length of
 $\sigma$, a \fs{} $\widehat{\sigma}\in\firseq{N^{\#}}{m}$, with 
 $m = \mathsf{m}\cup S^{\#}$, such that $X_{\widehat{\sigma}}\in \states{N^{\#}}$ and 
 $X = X_{\widehat{\sigma}}$.
 If $\lung{\sigma} = 0$ then $\lead{\sigma} = \mathsf{m}$ and it suffices to set
 $\widehat{\sigma} = \mathsf{m}\cup S^{\#}$. Assume it holds for $n$ and consider
 $\sigma = \sigma_1\trans{t}m$ of length $n+1$. To $\sigma_1$ in $\firseq{N}{\mathsf{m}}$ a 
 $\widehat{\sigma_1}$ in $\firseq{N^{\#}}{m}$ corresponds, so it is enough to
 verify that $\lead{\widehat{\sigma_1}}\trans{t}$. If $\neg \lead{\widehat{\sigma_1}}\trans{t}$
 it should be that a conflicting transition $t'$ has been executed in $\widehat{\sigma_1}$, but
 this implies that the conflicting transition has been executed in $\sigma$ as well, which 
 cannot be. Hence $\lead{\widehat{\sigma_1}}\trans{t}$ and the reached marking
 is $\lead{\sigma}\cup (S^{\#}\setminus \pre{X_{\sigma}})$, which implies that 
 $X = X_{\widehat{\sigma}}\in \states{N^{\#}}$. 
 $N^{\#}$ is then conflict saturated by construction, as the conflicting transitions in
 $N^{\#}$ are the same as in $N$, and it equivalent to $N$ as well.}
 \end{proof}
%
%We prove that $N^{\#}$ is a single execution net. Assume it is not, then there exists
% a transition $t\in T$ and a firing sequence $\sigma\in\firseq{N^{\#}}{m}$, with 
% $m = \mathsf{m}\cup S^{\#}$, such that $X_{\sigma}(t)>1$.  

A subnet of a net is a net obtained restricting places and transitions, and 
correspondingly restricting also the flow and the context relations as well as the initial marking
and the labeling.

 \begin{definition}
   \label{de:subnet-tran}
   Let $N = \langle S, T, F, I, R, \marko{m}, \ell\rangle$ be a net and let 
   \corr{$S'\subseteq S$ and}{} $T'\subseteq T$\corr{.}{~be a subset of transitions.}
   Then the subnet generated by $T'$ is the net
   $\subnet{N}{T'}{} = \langle S', T', F', I', R', \marko{m}', \ell'\rangle$, where
     \begin{itemize}
     \item
       $S' = \setcomp{s\in S}{\exists t\in T'.\ (s,t)\in F\ \lor\ (t,s)\in F}$,
     \item
       $F'$, $I'$ and $R'$ are the restriction of $F$, $I$ and $R$ to $S'$ and $T'$, 
     \item
       $\marko{m}'$ is the multiset on $S'$ obtained by $\marko{m}$ restricting to places in $S'$, and
      \item
       $\ell'$ is the restriction of $\ell$ to transitions in $T'$.
    \end{itemize}
 \end{definition}
\corr{}{The subnet is \emph{generated} by a subset $T'$ of transitions, 
hence the places to be considered are those connected with the transitions in $T'$.}

\section{Nets and event structures}\label{sec:netandes}
\aggiunta{In this section we revise and the more classical relations between nets and event structures,
in particular we recall the relationship between \emph{occurrence nets} and \emph{prime event structures}
and the one between \emph{unravel nets} and \emph{bundle event structures}. Though the latter one 
is not entirely \emph{classic} the notion of unravel net is closely related to the one of \emph{flow} net
and bundle event structures to \emph{flow} event structures.}

\subsection{Occurrence nets and prime event structure}
We recall the notion of \emph{occurrence} net, and as it has no inhibitor or read arc nor a labeling, 
we omit $I$, $R$ and $\ell$ in the following, assuming that $I = \emptyset = R$ and $\ell$ being the
identity on transitions. 
Given a net $N = \langle S, T, F, \mathsf{m}\rangle$,
we write $<_N$ for transitive closure of $F$.
We say $N$ is \emph{acyclic} if  $\leq_N$ is a partial order.
For occurrence nets, we adopt the usual convention: places and transitions are called 
as \emph{conditions} and \emph{events}, and use $B$  and 
$E$ for the sets of conditions and events.
We may confuse conditions with places and 
events with transitions.
The initial marking is denoted with $\mathsf{c}$.
\begin{definition}
 An \emph{occurrence net} (\cn) $\cnname = \langle B, E, F, \mathsf{c}\rangle$ is an acyclic, safe net 
 satisfying the
  following restrictions:
  \begin{itemize}
    \item 
       $\forall b\in B$. $\pre{b}$  
       is either empty or a singleton, and $\forall b\in \mathsf{c}$. $\pre{b} = \emptyset$, 
    \item 
       $\forall b\in B$. $\exists b'\in \mathsf{c}$ such that $b' \leq_{\cnname} b$,
    \item 
      for all $e\in E$ the set $\hist{e} = \setcomp{e' \in E}{e'\leq_{\cnname} e}$ is finite,
          and
   \item  
      $\#$ is an irreflexive and symmetric relation defined as follows:
           \begin{itemize}
             \item 
                $e\ \#_0\ e'$ iff $e, e' \in E$, $e\neq e'$ and 
                   $\pre{e}\cap\pre{e'}\neq \emptyset$,
             \item 
                $x\ \#\ x'$ iff $\exists y, y'\in E$ such that $y\ \#_0\ y'$ and $y \leq_{\cnname} x$ and 
                   $y' \leq_{\cnname} x'$.
             \end{itemize}
     \end{itemize}
\end{definition}
The intuition behind occurrence nets is the following: each condition $b$ represents 
the occurrence of a token,
which is produced by the \emph{unique} event in $\pre{b}$, 
unless $b$ belongs to the initial marking,
and it is used by only one transition (hence if $e, e'\in\post{b}$, then $e\ \#\ e'$).
On an occurrence net $\cnname$ it is natural to define a notion of \emph{causality} among elements of the 
net: we say that $x$ is \emph{causally dependent} on $y$ iff $y \leq_{\cnname} x$.
Occurrence nets are often the result of the \emph{unfolding} of a (safe) net.
In this perspective an occurrence net is meant to describe precisely the non-sequential semantics
of a net, and each reachable marking of the occurrence net corresponds to a reachable marking
in the net to be unfolded. 
Here we focus purely on occurrence nets and not on the nets they are the unfolding of.

\begin{proposition}
 Let $\cnname = \langle B, E, F, \mathsf{c}\rangle$ be an occurrence net. Then $\cnname$ is a single execution
 net and it is an unfolding.
\end{proposition}
\begin{proof}
 $\cnname$ is acyclic and safe, hence no transition in $E$ can be executed more than once, hence it is a single
 execution one, and it is an unfolding as configuration and states coincide, being the \emph{labelling}
 injective. 
\end{proof}
Occurrence nets are relevant as they are tightly related to \emph{prime event structures}, which
we briefly recall here (\cite{Win:ES}).

\begin{definition}\label{de:pes-winskel}
 A \emph{prime event structure ({\pes})} is a triple $\pesname = (E, <, \#)$, where 
 \begin{itemize}
  \item $E$ is a  countable set of \emph{events},
  \item $<\ \subseteq E\times E$ is
        an irreflexive partial order called the \emph{causality relation}, such that
        $\forall e\in E$. $\setcomp{e'\in E}{e' < e}$ is finite, and 
  \item $\#\ \subseteq E\times E$ is a \emph{conflict relation}, which is irreflexive,
        symmetric and \emph{hereditary} relation with respect to $<$: if  $e\ \#\ e' < e''$ 
        then $e\ \#\ e''$ for all $e,e',e'' \in E$.
 \end{itemize}       
\end{definition}
Given an event $e\in E$,  $\hist{e}$  denotes the set $\setcomp{e'\in E}{e'\leq e}$. 
A subset of events $X \subseteq E$ is left-closed if $\forall e\in X. 
\hist{e}\subseteq X$.
Given a subset $X\subseteq E$ of events,  $X$ is \emph{conflict free} iff 
for all $e, e'\in X$ it holds that 
$e\neq e'\ \Rightarrow\ \neg(e\ \#\ e')$, and we denote it with $\CF{X}$. 
Given $X\subseteq E$ such that $\CF{X}$ and $Y\subseteq X$, then also $\CF{Y}$.
\begin{definition}\label{de:pes-conf}
 Let $\pesname = (E, <, \#)$ be a {\pes}. Then $X\subseteq E$ is a \emph{configuration}
 if $\CF{X}$ and $\forall e\in X.\ \hist{e}\subseteq X$. 
 The set of configurations of the {\pes} $\pesname$  is denoted by $\Conf{\pesname}{\pes}$.
\end{definition}

%Configurations are definable also in occurrence nets.
%\begin{definition}\label{de:cn-conf}
% Let $\cnname = \langle B, E, F, \mathsf{c}\rangle$ be an \cn\ and 
% $X\subseteq E$ be a subset of events. Then $X$ is a \emph{configuration} of
% $\cnname$ whenever $\CF{X}$ and $\forall e\in X$. $\hist{e}\subseteq X$.
% The set of configurations of the \cn\ $\cnname$ is denoted by $\Conf{\cnname}{\cn}$.
%\end{definition}
%
%Given an \cn\ $\cnname = \langle B, E, F, \mathsf{c}\rangle$ and a state
%$X \in \states{\cnname}$, it is easy to see that it is \emph{conflict free}, \emph{i.e.}
%$\forall e, e'\in X$. $e\neq e'\ \Rightarrow\ \neg (e\ \#\ e')$, and \emph{left closed},
%\emph{i.e.} $\forall e \in X$. $\setcomp{e'\in E}{e'\leq_{\cnname} e}\subseteq X$.
%
%\begin{proposition}\label{pr:states-are-conf}
% Let $\cnname = \langle B, E, F, \mathsf{c}\rangle$ be an occurrence net and $X\in \states{\cnname}$. Then
% $X\in \Conf{\cnname}{\cn}$.
%\end{proposition}
%
Occurrence nets and prime event structures are connected as follows (\cite{Win:ES}).
Proofs are omitted as they are standard and can be found in literature.
\begin{proposition}\label{pr:on-to-pes}
 Let $\cnname = \langle B, E, F, \mathsf{c}\rangle$ be an \cn, and define $\nettoes{\cn}{\pes}(\cnname)$ 
 as the
 triple $(E, <_C, \#)$ where $<_C$ is the irreflexive and transitive relation obtained by $F$ 
 restricting to $E\times E$ and $\#$ is the irreflexive and symmetric relation associated to
 $\cnname$.
 Then $\nettoes{\cn}{\pes}(\cnname)$ is a \pes, and 
 $\Conf{\cnname}{\cn} = \Conf{\nettoes{\cn}{\pes}(\cnname)}{\pes}$.
\end{proposition}

Also the vice versa is possible, namely given a prime event structure one can associate to it an
occurrence net.
The construction is indeed quite standard (see \cite{Win:ES,BCP:LenNets} among many others).

\begin{definition}\label{de:pes-to-cn}
 Let $\pesname = (E, \leq, \#)$ be a \pes. Define 
 $\estonet{\pes}{\cn}(\pesname)$ as the net $\langle B, E, F, \mathsf{c}\rangle$ where
 \begin{itemize}
   \item $B = \setcomp{(\ast,e)}{e\in E}\cup \setcomp{(e,\ast)}{e\in E}\cup
         \setcomp{(e,e',<)}{e < e'} \cup \setcomp{(\setenum{e,e'},\#)}{e\ \#\ e'}$,
   \item $F =   \setcomp{(e,b)}{b = (e,\ast)}\ \cup\ \setcomp{(e,b)}{b = (e,e',<)}\ \cup\ 
                \setcomp{(b,e)}{b = (\ast,e)}\ \cup\ \setcomp{(b,e)}{b = (e',e,<)}\ \cup\ 
               \setcomp{(b,e)}{b = (Z,\#)\ \land\ e\in Z}$,    
         and
   \item $\mathsf{c} = \setcomp{(\ast,e)}{e\in E}\cup \setcomp{(\setenum{e,e'},\#)}{e\ \#\ e'}$.                                                
 \end{itemize}
\end{definition} 
\begin{proposition}\label{pr:pestoon-classic}
 Let $\pesname = (E, \leq, \#)$ be a \pes. Then  
 $\estonet{\pes}{\cn}(\pesname) = \langle B, E, F, \mathsf{c}\rangle$ as defined
 in Definition \ref{de:pes-to-cn} is an \cn, and 
 $\Conf{\pesname}{\pes} = \Conf{\estonet{\pes}{\cn}(\pesname)}{\cn}$
\end{proposition}

In essence an occurrence net is fully characterized by the partial order relation and the 
\emph{saturated} conflict relation. 
This observation, together with the fact that an immediate conflict in a safe net is represented
by a common place in the preset of the conflicting events, suggests that conflicts may be modeled
directly, which is the meaning of the following proposition\corr{ and that will be handy in rest of
the paper.}{.}

\aggiunta{\begin{proposition}
  Let $\cnname = \langle B, E, F, \mathsf{c}\rangle$ be an \cn\ and let $\#$ be the associated
  conflict relation. Then $\cnname^{\#} = \langle B\cup B^{\#}, E, F\cup F^{\#}, \mathsf{c}\cup B^{\#} \rangle$
  where $B^{\#} = \setcomp{\setenum{e,e'}}{e\ \#\ e'}$ and $F^{\#} = \setcomp{(A,e)}{A\in B^{\#}\ \land\ 
  e\in A}$, 
  is an \cn\ such that $\Conf{\cnname}{\cn} = \Conf{C^{\#}}{\cn}$.
\end{proposition}
\begin{proof}
 Along the same lines of Proposition~\ref{pr:sat-confl}.
\end{proof}}

\subsection{Unravel nets and bundle event structure}
We recall the notion of \emph{unravel} net (\cite{Pinna:PN11,CaPi:PN14,PF:jlamp20} or \cite{CP:soap17}). 
Similarly to \cn, also unravel nets do not have inhibitor arcs
or read arcs, and the labeling is injective, hence also here $I, R$ and $\ell$ are omitted.

 We say that a net $N = \langle S, T, F, \marko{m}\rangle$ is \emph{conflict-free} if
 $\forall s\in S$ it holds that $\post{s}$ is at most a singleton, and it is acyclic if the reflexive 
 and transitive closure of $F$ is a partial order over $S\cup T$.
 As we are considering safe nets we will confuse multisets with sets.
 \begin{definition}\label{de:unravel-net}
  An \emph{unravel net} (\un) $\unname = \langle S, T, F, \marko{m}\rangle$ is a safe net such that 
  \begin{itemize}
    \item for each state $X\in\states{\unname}$ the net $\subnet{\unname}{X}{}$ is acyclic and 
          conflict-free, and
    \item for each $t\in T$ there exists a state $X\in\states{\unname}$ such that $t\in\flt{X}$.
  \end{itemize}
 \end{definition} 
 Thus the whole unravel net is not constrained to be either acyclic or 
 conflict-free, but each of its executions
 gives an acyclic and conflict-free net. 
 \un{s} can be considered as the easiest generalization of occurrence nets, and indeed
 the following proposition shows that \cn{s} are \un{s} as well, as each execution of
 an \cn\ is clearly acyclic (as the whole \cn\ is acyclic) and conflict-free.
 The \emph{firability} of each event in an occurrence net is a consequence of the structural 
 constraint posed on this kind of net.
 
 \begin{proposition}\label{prop:causal-is-unravel}
  Let $\cnname$ be an \cn. Then $\cnname$ is an \un\ as well.
 \end{proposition}
 \aggiunta{
 \begin{proof}
   A \cn{} $\cnname$ is by definition acyclic, and if one consider a configuration, the subnet
   generated by the transitions (events) in this configuration gives a conflict-free one. Finally
   each transition (event) belongs to a configuration as each place (condition) is connected to
   an initially marked one. 
 \end{proof}
 }
 The followings propositions will be helpful in associating a \bes\ to an \un.
 \begin{proposition}\label{pr:s-m-empty}
   Let $\unname = \langle S, T, F, \mathsf{m}\rangle$ be an \un{} and
   let $s \in S$ such that $s\in \flt{\mathsf{m}}$ then $\pre{s} = \emptyset$.
 \end{proposition}
 \begin{proof}
   Assume $\pre{s} \neq \emptyset$, then there exists a transition $t \in T$ such that $t \in \pre{s}$. Now, 
   as $\unname$ is an unravel net, $t$ can be executed, hence there exists a \fs\ $\sigma$ such that 
   $\sigma\trans{t} m'$. We have two cases:
     \begin{enumerate}
       \item each transition $t'$ in $\sigma$ is such that 
             $s\notin \pre{t'}$, but then $m'(s) = 2$, contradicting the safeness of $\unname$, or
       \item there is a transition  $t'$ in $\sigma$ such that $s\in \pre{t'}$, but then 
             $\subnet{\unname}{\flt{X_\sigma}}{}$ is cyclic, 
             contradicting the hypothesis that $\unname$ is an unravel net.
             \qedhere
    \end{enumerate}  
  \end{proof}
 \begin{proposition}\label{pr:un-is-se}
   Let $\unname = \langle S, T, F, \mathsf{m}\rangle$ be an \un{}, then 
   $\unname$ is a single execution net.
 \end{proposition}
 \begin{proof}
   Assume it is not, then there exists a state 
   $X\in \states{\unname}$ and a transition $t\in T$ such that $X(t)>1$. This implies that
   there is a \fs\ $\sigma$ such that $\sigma = \sigma'\trans{t}\sigma''\trans{t}\sigma'''$ but 
   then, being $\unname$ a safe net, it should be that $\subnet{\unname}{\flt{X_\sigma}}{}$ 
   is cyclic as some transitions in $\flt{X_{\sigma''}}$ should put a token in a place
   $s\in \pre{t}$, contradicting the assumption that $\unname$ is an unravel net.
  \end{proof}

 Given an \un{} $\unname = \langle S, T, F, \mathsf{m}\rangle$, we define
 a \emph{semantic} conflict relation, and we say that $t\ \#\ t'$ iff $\forall X\in \states{\unname}$ it
 holds that $\setenum{t, t'}\not\subseteq\flt{X}$. \aggiunta{Note that this semantic conflict relation
 has already been implicitly used in showing how to construct an equivalent conflict saturated net.}

 \corr{The following proposition says that two transitions
 with a common place in their postset are in confict.}
 {The following two propositions say that two 
 transitions in an \un{} are in conflict if they have a common place either in their presets or
 in their postsets. The common place in the preset suggests that their conflict is immediate, the
 common place implies that only one of the two
 transitions can belong to the \emph{history} (firing sequence) marking that place, as the proof itself
 shows.}
 \begin{proposition}\label{s-tr-confl1}
   Let $\unname = \langle S, T, F, \mathsf{m}\rangle$ be an \un, let
   $s\in S$ be a place and $t,t'\in T$ be two transitions of the \un. 
   If $\post{t}\cap\post{t'} \neq \emptyset$
   then $t\ \#\ t'$.
  \end{proposition}
  \begin{proof}
    Let $s\in \post{t}\cap\post{t'}$. Assume that $\neg(t\ \#\ t')$, then either there exists a \fs\ $\sigma$   
    such that $\sigma = \sigma'\trans{t}\sigma''\trans{t'}\sigma'''$ or there exists 
    $\sigma = \sigma'\trans{t'}\sigma''\trans{t}\sigma'''$. Assume 
    $\sigma = \sigma'\trans{t}\sigma''\trans{t'}\sigma'''$\corr{, then 
    $\sigma''$ is either a marking $\bar{m}$  or it is 
    $\bar{m} \trans{\bar{t}} \bar{\sigma}''$.
    Then $\lead{\sigma'\trans{t}\bar{m}}(s) = 1$ and 
    $\lead{\sigma'\trans{t}\sigma''\trans{t'}\hat{m}}\geq 1$. If 
    $\lead{\sigma'\trans{t}\sigma''\trans{t'}\hat{m}}\geq 2$ the net is not safe, and if 
    $\lead{\sigma'\trans{t}\sigma''\trans{t'}\hat{m}} = 1$}{.~Take 
    the marking $\lead{\sigma'\trans{t}\sigma''}$. If 
    $\lead{\sigma'\trans{t}\sigma''}(s) = 1$ then firing $t'$ the place $s$ get marked twice,
    and if $\lead{\sigma'\trans{t}\sigma''}(s) = 0$ then }     
    the net $\subnet{\unname}{\flt{X_{\sigma}}}{}$ 
    is not acyclic, in both cases contradicting the assumption that $\unname$ is an unravel net.
  \end{proof}

 \begin{proposition}\label{s-tr-confl2}
   Let $\unname =\langle S, T, F, \mathsf{m}\rangle$ be an \un, and let $t,t'\in T$ be two transition 
   of the \un. 
   Then $\pre{t}\cap\pre{t'}\neq \emptyset$ implies that $t\ \#\ t'$. 
 \end{proposition}
 \begin{proof}
   Assume that $\neg(t\ \#\ t')$ and let $s\in \pre{t}\cap\pre{t'}$. 
   Then exists a \fs\ $\sigma$ such that $\sigma'\trans{t}\sigma''\trans{t'}\sigma'''$ 
   or $\sigma'\trans{t'}\sigma''\trans{t}\sigma'''$. Using the same argument of 
   Proposition~\ref{s-tr-confl1} we have that $s$ would get marked twice violating the acyclicity of 
   $\subnet{\unname}{\flt{X_\sigma}}{}$.
\end{proof}

We recall now the notion of \emph{bundle event structures}~\cite{Langerak:1992:BES} (\bes).
In a \bes{} causality is represented by pairs $(X,e)$, the \emph{bundles},
where $X$ is a non empty set of events and $e\not\in X$ an event. 
The meaning of a bundle $(X,{e})$ is that if ${e}$ happens then one (and only
one) event of $X$ has to have happened before (events in $X$ are
pairwise conflicting).  An event ${e}$ can be caused by several
bundles, in that case, for each bundle an event in it should have happened.

\begin{definition}\label{de:bes}
A \emph{bundle event structure} is a triple $\besname = (E,\bundle, \#)$, where 
\begin{itemize}
  \item ${E}$ is a set of events, 
  \item $\#$ is an irreflexive and symmetric binary relation on ${E}$ 
         (the \emph{conflict} relation),
  \item $\bundle\ \subseteq\ \Powfin{{E}}\times {E}$ is the \emph{enabling}
        relation such that if $X \bundle {e}$ then for all ${e}_1, {e}_2 \in X$. 
        ${e}_1 \neq {e}_2$ implies ${e}_1\ \#\ {e}_2$, and
  \item for each ${e}\in E$ it holds that the set $\bigcup \setcomp{X}{X \bundle {e}}$ is finite.    
\end{itemize}        
\end{definition}

The final condition is an analogous of the finite cause requirement for prime event structures. Indeed
this requirement rules out situations like the following one. Consider an event $s$ such that 
$\forall i\in \nat$ there is a bundle $\setenum{e_i}\bundle e$. Then the event $e$ has infinite causes
which we want to rule out.

The configurations of a \bes\ are defined as follows.

\begin{definition}\label{de:bes-conf}
 Let $\besname = ({E}, \bundle, \#)$ be a \bes\ and $X\subseteq {E}$ be a set
 of events. Then $X$ is a \emph{configuration} of $\besname$ iff 
 \begin{enumerate}
   \item it is \emph{conflict free}, \emph{i.e.} $\forall {e}, {e}'\in X$.
         ${e} \neq {e}'\ \Rightarrow\ \neg ({e}\ \# {e}')$, and
   \item there exists a linearization $\setenum{{e}_1, \dots, {e}_n, \dots}$ 
         of the events in $X$ such that $\forall i\in\nat$ and for all bundles
         $X_{j_i}\bundle {e}_i$ it holds that
         $X_{j_i}\cap \setenum{{e}_1, \dots, {e}_{i-1}} \neq \emptyset$.              
 \end{enumerate}
 The set of configurations of a \bes\ $\besname$ is denoted with $\Conf{\besname}{\bes}$.
\end{definition}
The requirements are the usual ones: it must be conflict free and each event must have 
all of its causes. The causes of an event in a \bes, as said before, have to be chosen using
all the bundles involving the event.
 
Clearly \bes\ are a conservative extension of \pes.
Indeed each \pes\ $P = ({E}, \leq, \#)$ can be seen as the \bes{} stipulating that
the bundles are $\bundle = 
\setcomp{(\setenum{{e'}},{e})}{{e'}<{e}}$. 
\bes\ are more expressive than \pes, as they are able to model
\emph{or}-causality.  In fact we may have the following
\bes\ ${a}\ \#\ {b}$ (symmetric pair omitted) and
$\setenum{{a}, {b}}\bundle {c}$ stipulating that the same
event may have two different and alternative pasts, namely one
containing ${a}$ and the other ${b}$.

Like \pes\ and  \cn, also \bes\ and \un{} are closely related.
The intuition is rather simple: to each place in the preset of a
transition $t$ we associate a bundle $X$ for the correspondent event
$t$ in the event structure, and the bundle is formed by the
transitions putting a token in that place.

\begin{proposition}\label{prop:un-bes}
 Let $\unname = \langle S, T, F, \marko{m} \rangle$ be an \un, then 
 $\nettoes{\un}{\bes}(\unname) = (T, \bundle, \#)$ is a \bes, where 
 \begin{itemize}
  \item  
    $t\ \#\ t'$ in $\nettoes{\un}{\bes}(\unname)$ iff  
        $t\ \#\ t'$ in $\unname$, and
  \item
    for each $t\in T$, for each $s\in \pre{t}$, we have 
    $\pre{s} \bundle t$.
 \end{itemize}
 Furthermore $\Conf{\unname}{\un} = \Conf{\nettoes{\un}{\bes}(\unname)}{\bes}$.
\end{proposition}
\begin{proof}
We show that $\nettoes{\un}{\bes}(\unname)$ is indeed an \bes.
The conflict relation in $\nettoes{\un}{\bes}(\unname)$ is antisymmetric and irreflexive
because it is so in $\unname$. Since we are are dealing with unravel nets, all
the transitions putting tokens in the same place are in conflict, so the
bundles respect Definition~\ref{de:bes}.
We show that $\Conf{\unname}{\un} = \Conf{\nettoes{\un}{\bes}(\unname)}{\bes}$. Consider
$X \in \Conf{\unname}{\un}$, then there is a \fs\ $\sigma$ such that
$X_\sigma = X$ and for each $i\in\mathit{dom}(\trace{\sigma})$ we have that
$\sigma(i-1)(s) = 1$ for each $s\in \pre{t_i}$. 
It is easy to see
that $\toset{\trace{\sigma}}$ is a configuration of $\nettoes{\un}{\bes}(\unname)$
as for each $i\in \mathit{dom}(\trace{\sigma})$ it holds that
$\toset{\trace{\sigma}(i-1)}\cap \pre{s}\neq\emptyset$ for each
$\pre{s}\bundle t_i$. Conflict-freeness is trivial.
\end{proof}

%\begin{example}
%Consider the net in Fig.~\ref{fig:simple-un} and its associated event
%structure \ref{fig:bes-2}. The bundles are $\setenum{{c}, {d}}\bundle {e}$
%and $\setenum{{a}, {b}}\bundle {e}$, and are obtained syntactically, whereas
%the conflicts ${c}\ \# {d}$, ${a}\ \# {b}$, ${a}\ \# {c}$ and 
%${b}\ \# {c}$ are
%obtained using all the possible executions, though in this example they are syntactically 
%deducible as they share a condition in their preset.
%\end{example}

\begin{proposition}\label{prop:bes-un}
 Let $\besname = (E, \bundle, \#)$ be a \bes\ such that
 $\forall {e}\in {E}$ $\exists X\in \Conf{\besname}{\bes}.\ e\in X$.
 Then $\estonet{\bes}{\un}(\besname) = \langle S, E, F, \marko{m} \rangle$ 
 where 
 \begin{itemize}
  \item $S = \setcomp{({e},\ast)}{{e} \in E} \cup \setcomp{\setenum{e,e'}}{{e}\ \#\ {e}'} \cup 
            \setcomp{(Y,e)}{Y \bundle {e}}\cup \setcomp{(\ast,e)}{{e} \in E}$,
  \item \(F = \setcomp{(s,e)}{s = (e,\ast) \lor s = (Y,e) \lor s = \setenum{e,e'}} \cup \setcomp{(e,s)}{s = (\ast,e) \lor
                      (s = (Y,{e}') \land e\in Y)}\), and
  \item $\marko{m} = \setcomp{({e},\ast)}{{e} \in E}\cup \setcomp{\setenum{e,e'}}{{e}\ \#\ {e}'}$
 \end{itemize}  
is an unravel net and $\Conf{\besname}{\bes} = \Conf{\estonet{\bes}{\un}(\besname)}{\un}$.
\end{proposition}

\begin{proof}
The safeness of $\estonet{\bes}{\un}(\besname)$ results from the fact that the
places $({e},i)$ allows only one execution of each transitions
${e}$ as they have no incoming arc, and the places with more that
one incoming arc, the $S_B$ ones, cannot be marked by more than one
transition because of the conflicts in the bundle set, finally the
places $\setenum{e,e'}$ does not allow the execution of
conflicting transitions in the same run.

Let $\sigma$ be a \fs\ of $\estonet{\bes}{\un}(\besname)$ and be 
$\trace{\sigma} = e_1e_2e_3\cdots e_n\cdots$ 
the associated trace. 
By induction on the indexes in $\mathit{dom}(\trace{\sigma})$ we show that 
$\subnet{\estonet{\bes}{\un}(\besname)}{X}{}$ where $X = X_{\sigma(i)}$ is a
conflict-free causal net.
\begin{itemize}
 \item The empty trace gives a trivial net (a net without places and transitions) which is
       vacuously a conflict-free causal net,
 \item consider the trace $\trace{\sigma(n)} {e}_1 \cdots {e}_{n-1}{e}_n$ and let 
       $\toset{\trace{\sigma(n)}}$ be the state $\setenum{{e}_1,\dots,{e}_{n-1}{e}_n}$ associated to 
       $\trace{\sigma(n)}$. As $\trace{\sigma(n-1)} = {e}_1 \cdots {e}_{n-1}$ is a trace as well, 
       by inductive hypothesis $\subnet{\estonet{\bes}{\un}(\besname)}{\toset{\trace{\sigma(n-1)}}}{} = 
       \langle S_{n-1}, E_{n-1}, F_{n-1}, \marko{m}_{n-1} \rangle$
       is a conflict-free causal net where 
       $S_{n-1} = \pre{\toset{\trace{\sigma(n-1)}}} \cup \post{\toset{\trace{\sigma(n-1)}}}$,
       $E_{n-1} = \toset{\trace{\sigma(n-1)}}$, $F_{n-1}$ is the restriction of $F$ to the
       transitions in $E_{n-1}$ and the places in $S_{n-1}$, and 
       $\marko{m}_{n-1}$ is the restriction of $\marko{m}$ to the places in $S_{n-1}$.
      Consider now that subnet $\subnet{\estonet{\bes}{\un}(\besname)}{\toset{\trace{\sigma(n)}}}{}$ 
      and assume   
      that a place $s$ in $\post{{e}_n}$ already belongs to $S_{n-1}$. Then $s$ must be 
      a place  of the kind $(Y,{e}_i)$, for some ${e}_i \in \toset{\trace{\sigma(n-1)}}$, 
      but then there must be a ${e}_j$ in $\toset{\trace{\sigma(n-1)}}$ such
      that ${e}_j\in Y$, which this contradicts the fact that $\toset{\trace{\sigma(n-1)}}$ is
      conflict-free as ${e}_j\ \#\ {e}_i$. So ${e}_n$ cannot mark places already
      marked in the past, and the subnet 
      $\subnet{\estonet{\bes}{\un}(\besname)}{\toset{\trace{\sigma(n)}}}{}$ 
      is a conflict-free causal net. 
\end{itemize}
The proof that $\Conf{\besname}{\bes} = \Conf{\estonet{\bes}{\un}(\besname)}{\un}$ goes along the
same reasoning of the previous proposition.
\end{proof}

\section{Context-Dependent Event Structure}\label{sec:es-mie}
We recall the notion of \emph{Context-Dependent} event structure introduced in \cite{Pi:coord19} and
further studied in \cite{Pi:lmcs}.
The idea is that the happening of an event depends on a set of modifiers (the \emph{context}) 
and on a set of \emph{real} dependencies, which are activated by the set of modifiers. 
\aggiunta{Context-Dependent event structures are characterized by a novel \emph{context-dependency}
relation which subsumes all the others dependency relations studied in literature.
In presenting them we follow closely the approach taken in \cite{Pi:coord19} and \cite{Pi:lmcs}, thus 
we present them in their full generality and then show that the \emph{context-dependency}
relation can have a suitable and simpler form, yielding the notion of \emph{elementary} \Ges.}

\begin{definition}\label{de:mia-es}
 A \emph{context-dependent event structure} ({\Ges}) is a triple 
 $\cdesname = (E, \#, \gesrel)$ where 
 \begin{itemize}
  \item 
  $E$ is a set of \emph{events},
  \item 
  $\#\ \subseteq E\times E$ is an irreflexive and symmetric relation, called 
        \emph{conflict relation}, and
  \item 
  $\gesrel\ \subseteq \Pow{\pmv{A}}\times E$,  
        where $\pmv{A} \subseteq \Powfin{E}\times\Powfin{E}$,   
        is a relation, called
        the \emph{context-dependency} relation (\cd-relation), which is such that for each 
        $\pmv{Z}\gesrel e$ it holds that 
         \begin{itemize}
          \item 
                $\pmv{Z}\neq \emptyset$ and $|\pmv{Z}|$ is finite,
          \item 
                for each $(X,Y)\in \pmv{Z}$ it holds that $\CF{X}$ and
                $\CF{Y}$, 
                and
          \item for each $(X,Y), (X',Y')\in \pmv{Z}$ if $X = X'$ then $Y = Y'$.
         \end{itemize}
 \end{itemize}
\end{definition}
The \cd-relation models, for each event, which are the possible contexts in which the event may happen 
(the first component of each pair) and for each context which are the events that have to be occurred 
(the second component).
\aggiunta{The idea conveyed by this relation is originated by the possibility that dependencies of 
an event may vary, \eg{} growing or shrinking depending on the happening of suitable modifiers, like in
the event structures presented in \cite{AKPN:DC,AKPN:lmcs18}, or like what happen in \cite{GP:ESRC} where
conflicts can be resolved. The \cd-relation is capable of modeling this dynamicity.
We will later see that this dynamicity can be  represented in suitable nets with contextual arcs.~}
The requirement posed on the $\gesrel$ are rather few. The first one is that
for $\pmv{Z}\gesrel e$ the set $Z$ is finite and not empty. The finiteness requirement
mimics the usual requirement that an event has a finite set of causes, whereas 
the non emptiness is justified by the fact that an event need a context to happen,
and though the context may be the empty one, it should in any case be present.
The second is again rather obvious, contexts must be conflict-free subset of events.
For the last one the intuition is that \corr{once determined a context, the other component should be
unique.}{for each context, the other component should be unique.} 
\aggiunta{The few constrains posed on the $\gesrel$ relation have the drawback that this relation
can be less \emph{informative}.}

We now recall the notion of enabling of an event. 
We have to determine, for each $\pmv{Z} \gesrel e$, which of the contexts
$X_i$ should be considered. To do so we define the \emph{context} associated to each 
entry of the \cd-relation.
Given $\pmv{Z} \gesrel e$, where $\pmv{Z} = \setenum{(X_1,Y_1), \dots, (X_n,Y_n)}$, 
with $\ctx(\pmv{Z})$ we
denote the set of events $\bigcup_{i=1}^{|Z|} X_i$, and this is the one regarding 
$\pmv{Z} \gesrel e$.

\begin{definition}\label{de:mia-es-enabling}
 Let $\cdesname = (E, \#, \gesrel)$ be a \Ges\ and $C\subseteq E$ be a subset of events. 
 Then
 the event $e\not\in C$ is \emph{enabled} at $C$, denoted with $C\enab{e}$, 
 if for each $\pmv{Z} \gesrel e$, with 
 $\pmv{Z} = \setenum{(X_1,Y_1), \dots, (X_n,Y_n)}$, there is a pair $(X_i,Y_i)\in \pmv{Z}$ such that
 $\ctx(\pmv{Z})\cap C = X_i$ and $Y_i\subseteq C$.
\end{definition}
Observe that requiring the non emptiness of the set $\pmv{Z}$ in $\pmv{Z} \gesrel e$ guarantees that 
an event $e$ may be enabled at some subset of events.

\begin{definition}
\label{de:mia-es-event-traces}
Let $\cdesname = (E, \#, \gesrel)$ be a \Ges. 
Let $C$ be a subset of $E$. 
We say that $C$ is a \emph{configuration} of the \Ges\ $\cdesname$ 
iff there exists a 
sequence of distinct events $\rho = e_1e_2\cdots$ over $E$ 
such that 
\begin{itemize}
  \item 
  $\toset{\rho} = C$,
  \item 
  $\toset{\rho}$ is conflict-free, and 
  \item 
  $\forall 1 \leq i \leq \len{\rho}.\ \toset{\rho}_{i-1} \enab{e_i}$.
\end{itemize}
With $\Conf{\cdesname}{\Ges}$ we denote the set of configurations of the \Ges\ $\cdesname$.
\end{definition}
We illustrate this kind of event structure with some examples, mainly taken from \cite{Pi:coord19} and
\cite{Pi:lmcs}.

\begin{example}\label{ex:new2res}
 Consider three events $\pmv{a}, \pmv{b}$ and $\pmv{c}$. All the events are 
 singularly enabled but $\pmv{a}$ and $\pmv{b}$ are in conflict unless $\pmv{c}$ has not happened
 (we will see later that this are called \emph{resolvable} conflicts).
 Hence for the event $\pmv{a}$ we stipulate 
 $$\setenum{(\emptyset,\emptyset),\corrmath{(\setenum{\pmv{c}},\emptyset),}{}(\setenum{\pmv{b}},\setenum{\pmv{c}})}
 \gesrel \pmv{a}$$
 that should be interpreted as follows: if the context is $\emptyset$ \corr{or $\setenum{\pmv{c}}$}{} then
 $\pmv{a}$ is enabled without any further condition (the $Y$ are the empty set), if the context
 is $\setenum{\pmv{b}}$ then also $\setenum{\pmv{c}}$ should be present.
 The set $\ctx(\setenum{(\emptyset,\emptyset),\corrmath{(\setenum{\pmv{c}},\emptyset),}{}
 (\setenum{\pmv{b}},\setenum{\pmv{c}})})$ is $\setenum{\pmv{b}}$. 
 Similarly, for the event
 $\pmv{b}$ we stipulate 
 $$\setenum{(\emptyset,\emptyset),\corrmath{(\setenum{\pmv{c}},\emptyset),}{}(\setenum{\pmv{a}},\setenum{\pmv{c}})}\gesrel \pmv{b}$$ 
 which is justified as above and finally for the event $\pmv{c}$ 
 we stipulate 
 $$\setenum{(\emptyset,\emptyset)\corrmath{,(\setenum{\pmv{a}},\emptyset),(\setenum{\pmv{b}},\emptyset)}{}}
 \gesrel \pmv{c}$$
 namely any context allows to add the event.
 \aggiunta{The empty set $\emptyset$ is a configuration, and so are the singletons 
 $\setenum{\pmv{a}}$, $\setenum{\pmv{b}}$ and 
 $\setenum{\pmv{c}}$ which are reached from $\emptyset$ as $\pmv{a}, \pmv{b}$ and $\pmv{c}$ are 
 enabled at this configuration, $\setenum{\pmv{a}, \pmv{c}}$ is a configuration and can be reached
 from $\setenum{\pmv{c}}$ adding $\pmv{a}$ or from $\setenum{\pmv{a}}$ adding $\pmv{c}$, and
 analogously $\setenum{\pmv{b}, \pmv{c}}$ is a configuration, and finally 
 $\setenum{\pmv{a}, \pmv{b}, \pmv{c}}$ is a configuration that can be reached from 
 $\setenum{\pmv{a}, \pmv{c}}$ adding $\pmv{b}$ or from $\setenum{\pmv{b}, \pmv{c}}$ adding $\pmv{a}$.
 $\setenum{\pmv{a}, \pmv{b}}$ it is not a configuration as it cannot be reached from 
 $\setenum{\pmv{a}}$ adding $\pmv{b}$ and also from $\setenum{\pmv{b}}$ adding $\pmv{c}$. In fact,
 to add $\pmv{a}$ to $\setenum{\pmv{b}}$ one has the context $\setenum{\pmv{b}}$ but
 $\setenum{\pmv{c}}$ is not contained in $\setenum{\pmv{b}}$.}
\end{example}

\begin{example}\label{ex:new2dim}
 Consider three events $\pmv{a}, \pmv{b}$ and $\pmv{c}$, and assume that $\pmv{c}$ depends on 
 $\pmv{a}$ unless the event $\pmv{b}$ has occurred, and in this case this dependency is removed. 
 Thus there is a classic causality 
 between $\pmv{a}$ and $\pmv{c}$, but it can dropped if $\pmv{b}$ occurs. Clearly $\pmv{a}$ and
 $\pmv{b}$ are always enabled. The \cd-relation is 
 $\setenum{(\emptyset,\emptyset)}\gesrel \pmv{a}$, 
 $\setenum{(\emptyset,\emptyset)}\gesrel \pmv{b}$ and
 $\setenum{(\emptyset,\setenum{\pmv{a}}),(\setenum{\pmv{b}},\emptyset)}\gesrel \pmv{c}$.
 \aggiunta{In this case the configurations are $\emptyset$, $\setenum{\pmv{a}}$, $\setenum{\pmv{b}}$,
 $\setenum{\pmv{a},\pmv{c}}$ (reachable from $\setenum{\pmv{a}}$), $\setenum{\pmv{b},\pmv{c}}$ (reachable
 from $\setenum{\pmv{b}}$), $\setenum{\pmv{a},\pmv{b}}$ and $\setenum{\pmv{a},\pmv{b},\pmv{c}}$ (reachable 
 from $\setenum{\pmv{a},\pmv{b}}$, $\setenum{\pmv{b},\pmv{c}}$ and $\setenum{\pmv{a},\pmv{c}}$).}
\end{example}

\begin{example}\label{ex:new2aug}
 Consider three events $\pmv{a}, \pmv{b}$ and $\pmv{c}$, and assume that $\pmv{c}$ depends on 
 $\pmv{a}$ just when the event $\pmv{b}$ has occurred, and in this case this dependency is added,
 otherwise it may happen without.  
 Thus classic causality relation
 between $\pmv{a}$ and $\pmv{c}$ is added if $\pmv{b}$ occurs. Again $\pmv{a}$ and
 $\pmv{b}$ are always enabled. The \cd-relation is 
 $\setenum{(\emptyset,\emptyset)}\gesrel \pmv{a}$, 
 $\setenum{(\emptyset,\emptyset)}\gesrel \pmv{b}$ and
 $\setenum{(\emptyset,\emptyset),(\setenum{\pmv{b}},\setenum{\pmv{a}})}\gesrel \pmv{c}$.
 \aggiunta{The configurations of this \Ges{} are $\emptyset$, $\setenum{\pmv{a}}$, $\setenum{\pmv{b}}$,
 $\setenum{\pmv{c}}$, $\setenum{\pmv{a},\pmv{c}}$ (reachable from $\setenum{\pmv{a}}$), 
 $\setenum{\pmv{a},\pmv{b}}$ and $\setenum{\pmv{a},\pmv{b},\pmv{c}}$, which is for instance reachable 
 from $\setenum{\pmv{a},\pmv{b}}$ adding $\pmv{c}$. The configuration $\setenum{\pmv{b},\pmv{c}}$ is
 reachable from $\setenum{\pmv{c}}$ adding $\pmv{b}$ but not from the configuration $\setenum{\pmv{b}}$.}
\end{example}
These examples should clarify how the \cd-relation is used and also that each event
may be \emph{implemented} by a different pair $(X,Y)$ of modifiers and dependencies.

In \cite{Pi:coord19} and \cite{Pi:lmcs} we have shown that many event structures can be seen
as a \Ges, and this is obtained taking the configurations of an event structure and from
these synthesizing the conflict and the \gesrel\ relations. 
%

%\centerline{\Large parte nuova}

To characterize how the configurations of a \Ges\ are organized
we recall the notion of \emph{event automaton} \cite{PP:NEPC}.

\begin{definition}\label{de:event-automata}
 Let $E$ be a set of \emph{events}. An \emph{event automaton} over $E$ (\ea) 
 is the tuple $\eaname = \tuple{E, \stati, \earel, \instate}$ such that
 \begin{itemize}
   \item 
   $\stati \subseteq \Pow{E}$, and
   \item 
   $\earel \subseteq \stati\times\stati$ is such that $s\earel s'$ implies that
         $s \subset s'$.
 \end{itemize}
 $\instate\in \stati$ is the initial state.
\end{definition}
An event automaton is just a set of subsets of events and a reachability relation $\earel$ with the
minimal requirements that if two states $s, s'$ are related by the $\earel$ relation, namely
$s\earel s'$, then $s'$ is \emph{reached} by $s$ adding at least one event.

In \cite{Pi:coord19,Pi:lmcs} we have proven that, given a \Ges\ 
$\cdesname = (E, \#, \gesrel)$, the quadruple $\tuple{E, \stati, \earel, \instate}$,
where $\stati = \Conf{\cdesname}{\Ges}$, 
$\earel \subseteq \Conf{\cdesname}{\Ges}\times \Conf{\cdesname}{\Ges}$
is such that $C \earel C'$ iff $C'\setminus C = \setenum{e}$ and $\instate = \emptyset$, is an
\ea. The event automaton associated to $\cdesname$ is denoted with $\cdestoea(\cdesname)$.

With the aid of \ea\ we can establish when two \Ges\ are \emph{equivalent}.

\begin{definition}\label{de:ges-equiv}
 Let $\cdesname = (E, \#, \gesrel)$ and $\cdesname' = (E', \#', \gesrel')$ be two \Ges. 
 Then $\cdesname$ is \emph{equivalent} to $\cdesname'$, written as $\cdesname \cong \cdesname'$,
 whenever $E = E'$ and $\cdestoea(\cdesname) = \cdestoea(\cdesname')$.
\end{definition}

We list some property an \ea\ may enjoy. 

\begin{definition}\label{de:ea-simple}
 Let $\eaname = \tuple{E, \stati, \earel, \instate}$ be an \ea. We say
 that $\eaname$ is \emph{simple} if 
      $\forall e\in E$. $\exists s\in \stati$ such that 
      $s\cup\setenum{e}\in \stati$, 
      $s\earel s\cup\setenum{e}$ and $s\in \Powfin{E}$.
\end{definition}
In a simple event automaton, for each event, there is a finite state 
such that this can be reached by adding just this event.

On states of an \ea\ we can define an operator
\corr{$\ragg{X} = \setcomp{s'\in\stati}{\exists s\in\stati.\ s\earel s'}$}
{$\ragg{X} = \setenum{\instate}\cup\setcomp{s'\in\stati}{\exists s\in X.\ s\earel s'}$}, and this is clearly
a monotone and continuous operator on subset of states, hence we can calculate
the least fixed point of it, namely $\mathit{lfp}(\raggname)$

\begin{definition}\label{de:ea-compl}
 Let $\eaname = \tuple{E, \stati, \earel, \instate}$ be an \ea. We
 say that the event automaton $\eaname$ is \emph{complete} iff 
 \corr{$\mathit{lfp}(\ragg{\setenum{\instate}}) = \stati$}{$\mathit{lfp}(\raggname) = \stati$}.
\end{definition}
In a complete \ea\ each state is reachable from the initial one.

\begin{definition}\label{de:ea-fin-caused}
 Let $\eaname = \tuple{E, \stati, \earel, \instate}$ be an \ea. We say
 that $\eaname$ is \emph{finitely caused} if 
 $\forall {e}\in E$. $\exists \eaallowset{\eaname}{e} = \setenum{X_1, \dots, X_n}$ 
 such that each $X_i\in\Powfin{E}$ and $\forall s\in\stati$. if
 $s\earel s\cup\setenum{e}$ then $\exists X\in\eaallowset{\eaname}{e}$ and $X\subseteq s$.
\end{definition}
In a finitely caused \ea\ each event that can be added to a state has a finite number
of justifying subsets, which again resemble a kind of finite causes principle.
We observe that in a finitely caused \ea\ $\eaname$ the sets $\eaallowset{\eaname}{e}$ can
be effectively calculated.

\aggiunta{Beside the notion of finitely caused we have to stipulate that an event cannot be \emph{inhibited}
in different infinite ways: we want to rule out the case that
the event $e$ is such that $\pmv{Z}\gesrel e$ with $\pmv{Z}$ infinite, as it would be, for instance,
in an \ea{} with states $\emptyset, \bigcup_{i\in\nat}\setenum{e_i}$ and 
$\bigcup_{i\in\nat}\setenum{e_i,e}$ with $\earel$ defined as $\emptyset \earel \setenum{e_i}$ and
$\setenum{e_i}\earel \setenum{e_i,e}$. In this \ea{} the event $e$ can be added to each singleton state,
and the unique context allowing it would be the empty one, and these states are infinite.}
\begin{definition}\label{de:ea-fin-inh}
 Let $\eaname = \tuple{E, \stati, \earel, \instate}$ be an \ea. We say
 that $\eaname$ is \emph{finitely inhibited} if 
 $\forall {e}\in E$ the set $\mathcal{I}(\eaname,e) =
 \setcomp{s\in \stati}{e\not\in s\ \land\ \exists s'\in\stati.\ s\cup\setenum{e}\subseteq s'\
 \land\ \forall s'\in\stati.\ \corrmath{}{e\in s'\ \Rightarrow\ }\neg (s\earel s')}$
 is finite.
\end{definition}
The intuition behind a finitely inhibited \ea\ $\eaname$ is the \aggiunta{indeed the~} \corr{following}{one hinted before}:
the number of states $s$ where an event $e$ cannot be added but there is another 
containing both the events in $s$ and $e$ \corr{are}{is} finite. 

\begin{proposition}\label{pr:cdesea-is-good}
 Let $\cdesname = (E, \#, \gesrel)$ be a \Ges\ and let $\cdestoea(\cdesname)$ be the
 associated \ea. Then $\cdestoea(\cdesname)$ is simple, complete, finitely caused and
 finitely inhibited.
\end{proposition}
\begin{proof}
 $\cdestoea(\cdesname)$ is simple as each event $e$
 belongs to a configuration $C$ and to $C$ a sequence $\rho = e_1e_2\cdots$ is associated.
 Assume $e$ is $e_i$ then $\toset{\rho_{i-1}}$ is a state and 
 $\toset{\rho_{i-1}} \earel \toset{\rho_{i-1}}\cup\setenum{e}$.
 Completeness depends on the fact that each configuration is reachable.
 The fact that it is finitely caused depends on the finiteness of $\pmv{Z}$.
 The finiteness of $\pmv{Z}$ for each event in the $\gesrel$ relation implies also that
 $\cdestoea(\cdesname)$ is finitely inhibited as otherwise there would be infinite states
 where an event can be added and these states have nothing in common, which contradicts
 the fact that $\cdestoea(\cdesname)$ is finitely caused.
\end{proof}

On the events of an \ea\ it is easy to define an irreflexive and symmetric conflict relation.
\begin{definition}\label{de:ea-confl}
 Let $\eaname = \tuple{E, \stati, \earel, \instate}$ be an \ea. 
 We define a symmetric and irreflexive conflict relation $\#_{\ea}$ as follows:
 $e\ \#_{\ea}\ e'$ iff for each $s\in \stati$.\  
 $\setenum{e, e'}\not\subseteq s$.
\end{definition}

We show how to associate to an \ea\ a \Ges. 

\begin{theorem}\label{th:ea-to-Ges}
 Let $\eaname = \tuple{E, \stati, \earel, \instate}$ be a simple, complete, finitely caused and
 finitely inhibited \ea\ 
 such that $E = \bigcup_{s\in\stati} s$. 
 Then $\toges{\ea}{\eaname} = 
 (E, \#, \gesrel)$ 
 is a \Ges, where $\#$ is the relation $\#_{\ea}$ of Definition~\ref{de:ea-confl},
 and for each $e\in E$ we have 
 $\setcomp{(X,\emptyset)}{X\in \eaallowset{\eaname}{e}}\cup
 \setcomp{(X,\setenum{e})}{X\in  \eainibset{\eaname}{e}}\gesrel e$,
 and $\cdestoea(\toges{\ea}{\eaname}) = \eaname$.
\end{theorem}
\begin{proof}
 Along the same line of Theorem 4.12 of \cite{Pi:lmcs}.
\end{proof}

%\centerline{\Large fine parte nuova}

We show here that it is possible to obtain the \gesrel\ relation where 
for each event $e$ there is just one entry $\pmv{Z} \gesrel e$.

\begin{definition}\label{de:normal-cdes}
 Let $\cdesname = (E, \#, \gesrel)$ be a \Ges. We say that $\cdesname$ is \emph{elementary}
 if $\forall e\in E$ there is just one entry $\pmv{Z} \gesrel e$.
\end{definition}

The \Ges\ in the Examples~\ref{ex:new2res},~\ref{ex:new2aug} and~\ref{ex:new2dim} are elementary ones.

\begin{theorem}
 Let $\cdesname = (E, \#, \gesrel)$ be a \Ges. Then there exists a 
 elementary \Ges\ $\cdesname' = (E, \#', \gesrel')$ such that
 $\Conf{\cdesname}{\Ges} = \Conf{\cdesname'}{\Ges}$.
\end{theorem}
\begin{proof}
 It is enough to observe that the \Ges\ obtained in Theorem~\ref{th:ea-to-Ges} is an 
 elementary one.
\end{proof}

We assume that the \Ges\ we will consider are elementary ones.

\section{Causal nets}\label{sec:causal-nets}
We introduce a notion that will play the same role of occurrence net or unravel net when related to
context-dependent event structure. 

Given a contextual Petri net $N = \langle S, T, F, I, R, \mathsf{m}, \ell\rangle$, 
we can associate to it a relation on
transitions, denoted with $\prec_{N}$ and defined as $t \prec_{N} t'$ when 
$\pre{t}\cap\inib{t'}\neq\emptyset$ or $\post{t}\cap \rarc{t'}\neq\emptyset$, with the aim of establishing
the \emph{dependencies} among transitions related by inhibitor or read arcs.
Similarly we can introduce a conflict relation among transitions, which is a \emph{semantic} one.
For this is enough to stipulate that two transitions $t, t'\in T$ are in conflict, denoted
with $t\ \#_{N}\ t'$ if $\forall X\in \states{N}.\ \setenum{t, t'}\not\subseteq\flt{X}$.
With the aid of these relations we can introduce the notion of pre-\emph{causal} net.

\begin{definition}\label{de:pre-causal-net}
 Let $N = \langle S, T, F, I, R, \mathsf{m}, \ell\rangle$ be a labeled Petri net over the set
 of label $\Lab$. Then $N$ is a pre-\emph{causal} net (p\ca\ net) if the following further conditions are
 satisfied:
 \begin{enumerate}
   \item $<_N \cap (T\times T) = \emptyset$, $\forall t\in T$. $\pre{t}\cap\inib{t} = \emptyset$ and 
         $\post{t}\cap\rarc{t} = \emptyset$,
   \item $\forall t\in T.\ \forall s\in\inib{t}.\ |\ell(\post{s})| = 1$,
   \item $\forall t, t'\in T$, $t \prec_{N} t'\ \Rightarrow t' \not\prec_{N} t$,
   \item $\forall t\in T$ the set $\inib{t}\cup\rarc{t}$ is finite,
   \item $\forall t, t'\in T.\ t\ \#_{N}\ t'\ \Rightarrow\ \pre{t}\cap\pre{t'}\neq\emptyset$, and
   \item $\forall t, t'\in T.\ (t\neq t'\ \land\ \ell(t) = \ell(t'))\ \Rightarrow\ t\ \#_{N}\ t'$.
%
%%   \item $\forall a\in\Lab.\ \forall t, t'\in \ell^{-1}(a).\ 
%%         \ell(\pre{(\inib{t})}\cup\post{(\inib{t})}) = \ell(\pre{(\inib{t'})}\cup\post{(\inib{t'})})$,
%   \item $\forall X\in \states{U}$ $\prec_{U}^{\ast}$ is a partial order on $X$, and
%   \item $\forall C\in \Conf{U}{\ca}.\ C = \flt{C}$
 \end{enumerate}  
\end{definition}
The first requirement implies that $\forall t, t'\in T$ we have that $\post{t}\cap\pre{t}' = \emptyset$,
hence in this kind of net the dependencies do not arise from the flow relation,
furthermore inhibitor and read arcs do not interfere with the flow relation. 
The second condition implies that if a token in a place inhibits the happening of a transition, then
all the transitions removing this token have the same label, 
the third is meant to avoid cycles between transitions arising from inhibitor and read arcs, 
the fourth one implies that for each transition $t$ the set $\setcomp{t'\in T}{t'\prec_U t}$ is finite,
the fifth one stipulates that two conflicting transitions (which never appear together in any execution of
the net) consume the same token from a place, and the last one that two different
transitions bearing the same label are in conflict.

\aggiunta{\begin{example}\label{ex:precn}
 Consider the net below.
 It is a p\ca{}. Condition (2) of the definition is fulfilled as, for instance, considering
$\inib{t_4} = \setenum{s_3,s_1}$, we have that all the transitions in $\post{s_1}$ have the same label,
and similarly for the unique transition in $\post{s_3}$.
Condition (3) is fulfilled as we have, \eg, 
$t_1 \prec t_3$ as $\post{t_3}\cap\rarc{t_1} \neq \emptyset$,
$t_1 \prec t_4$ as $\pre{t_1}\cap\inib{t_4} \neq \emptyset$, and
$t_2 \prec t_4$ as $\pre{t_2}\cap\inib{t_4} \neq \emptyset$, but never the vice versa.
Also the other conditions are easily checked.

 \begin{center}
  \scalebox{0.9}
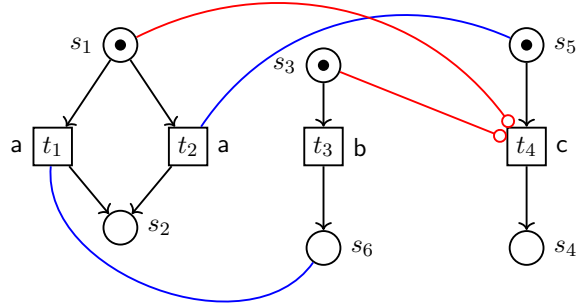
{\begin{tikzpicture}
\tikzstyle{inhibitor}=[o-,thick]
\tikzstyle{inhibitorred}=[o-,draw=red,thick]
\tikzstyle{inhibitorblu}=[o-,draw=blue,thick]
\tikzstyle{inhibitorpur}=[o-,draw=purple,thick]
\tikzstyle{pre}=[<-,thick]
\tikzstyle{post}=[->,thick]
\tikzstyle{read}=[-,thick]
\tikzstyle{readred}=[-,draw=red,thick]
\tikzstyle{readblu}=[-,draw=blue,thick]
\tikzstyle{readpur}=[-,draw=purple,thick]
\tikzstyle{transition}=[rectangle, draw=black,thick,minimum size=5mm]
\tikzstyle{invtransition}=[rectangle, draw=black!0,thick,minimum size=6mm]
\tikzstyle{place}=[circle, draw=black,thick,minimum size=5mm]
\node[place,tokens=1] (p1) at (1,3.5) [label=left:$s_1$] {};
\node[place,tokens=1] (p5) at (4,3.2) [label=left:$s_3$] {};
\node[place,tokens=1] (p3) at (7,3.5) [label=right:$s_5$]{};
\node[place] (p2) at (1,0.8) [label=right:$s_2$] {};
\node[place] (p4) at (7,0.5) [label=right:$s_4$]{};
\node[place] (p6) at (4,0.5) [label=right:$s_6$]{};
\node[transition] (t1a) at (0,2) [label=left:$\pmv{a}$] {$t_1$}
edge[pre] (p1)
edge[readblu, bend right=75] (p6)
edge[post] (p2);
\node[transition] (t2a) at (2,2) [label=right:$\pmv{a}$] {$t_2$}
edge[pre] (p1)
edge[readblu, bend left=40] (p3)
edge[post] (p2);
\node[transition] (tb) at (4,2) [label=right:$\pmv{b}$] {$t_3$}
edge[pre] (p5)
edge[post] (p6);
\node[transition] (t2c) at (7,2) [label=right:$\pmv{c}$] {$t_4$}
edge[pre] (p3)
edge[inhibitorred] (p5)
edge[inhibitorred, bend right=40] (p1)
edge[post] (p4);
\end{tikzpicture}}
 \end{center}
\end{example}}

Based on the notion of p\ca\ we can introduce the one of \emph{causal} net.
\begin{definition}\label{de:causal-net}
 Let $\caname = \langle S, T, F, I, R, \mathsf{m}, \ell\rangle$ be a pre-\emph{causal} net over the set
 of label $\Lab$. Then $\caname$ is a \emph{causal} net (\ca\ net) if
 \begin{enumerate} 
   \item $\forall X\in \states{K}$ $\prec_{K}^{\ast}$ is a partial order on $X$, and
   \item $\forall t\in T\ \exists X\in \states{K}.\  t\in \flt{X}$.
 \end{enumerate}
\end{definition}
The added conditions with respect to the ones in Definition~\ref{de:pre-causal-net} 
implies that the executions of the transitions in a \fs\ is compatible with
the dependency relation when the transitions in the state associated to the \fs\ are considered, and
each transition can be executed.
\aggiunta{The p\ca{} in Example~\ref{ex:precn} is a \ca{} 
as each transition can be executed and the transitive
closure of the $\prec$ relation is a partial order on states.}

\corr{I}{From the previous example i}t should be clear that the conditions posed on pre-causal nets 
and causal nets 
are meant to mimic some of the conditions
posed on an occurrence nets \aggiunta{(\eg{} finiteness of causes)~}or on similar one, like for 
instance \emph{unravel} nets or
flow nets (\cite{Bou:FESFN})\aggiunta{, \eg{} that when considering an execution, an order can be 
found among the various transitions}, and they should assure that it is comprehensible what a computation in
such a net can be looking at labels, as the main intuition is that for the same activity (label)
there may be several incarnations. \aggiunta{In fact, differently from \cn{s} and \un{s}, the labeling
mapping is not necessarily the identity and if two transitions share a place, then they are in conflict.}

\begin{example}\label{ex:causalnet}
The following one is a causal net:\medskip

\begin{center}
\begin{tikzpicture}
\tikzstyle{inhibitorred}=[o-, draw=red,thick]
\tikzstyle{inhibitorblu}=[o-, draw=blue,thick]
\tikzstyle{pre}=[<-,thick]
\tikzstyle{post}=[->,thick]
\tikzstyle{read}=[-,thick]
\tikzstyle{transition}=[rectangle, draw=black,thick,minimum size=5mm]
\tikzstyle{place}=[circle, draw=black,thick,minimum size=5mm]
\node[place,tokens=1] (p1) at (0,2) {};
\node[place,tokens=1] (p3) at (3,2) {};
\node[place,tokens=1] (p5) at (6,2) {};
\node[place] (p2) at (0,0) {};
\node[place] (p4) at (3,1) {};
\node[place] (p6) at (6,0) {};
\node[transition] (t1) at (0,1) [label=left:$\pmv{b}$] {$t_1$}
edge[pre] (p1)
edge[post] (p2);
\node[transition] (t2) at (1.5,1) [label=left:$\pmv{c}$] {$t_2$}
edge[pre] (p3)
edge[inhibitorred] (p1)
edge[post] (p4);
\node[transition] (t3) at (4.5,1) [label=right:$\pmv{c}$] {$t_3$}
edge[pre] (p3)
edge[read] (p6)
edge[inhibitorred,bend left] (p2)
edge[post] (p4);
\node[transition] (t4) at (6,1) [label=right:$\pmv{a}$] {$t_4$}
edge[pre] (p5)
edge[post] (p6);
\end{tikzpicture}
\end{center}
\medskip 

All the conditions of Definitions~\ref{de:pre-causal-net} and \ref{de:causal-net} are fulfilled. 
The two transitions bearing the same
label ($t_2$ and $t_3$) are conflicting ones, namely they never appear together in any computation
though the activity realized by these two transitions ($\pmv{c}$) appears in all maximal computations.
\end{example}
The first observation we make on causal nets is that they are good candidates to be
seen as a \emph{semantic} net, namely a net meant to represent the behaviour of a system properly
modeling dependencies and conflicts of any kind. 

\begin{proposition}\label{pr:ca-net-consat}
 Let $\caname = \langle S, T, F, I, R, \mathsf{m}, \ell\rangle$ be a causal net. Then
 $\caname$ is conflict saturated.
\end{proposition}
\begin{proof}
 A pre-causal net is conflict saturated by construction, as it is required that
 two conflicting transitions share a place in their presets.
\end{proof}

\begin{proposition}\label{pr:ca-single-ex-net}
 Let $\caname = \langle S, T, F, I, R, \mathsf{m}, \ell\rangle$ be a causal net. Then
 $\caname$ is an single execution net.
\end{proposition}
\begin{proof}
 Obvious, as $\forall t, t'\in T$ it holds that $\post{t}\cap\pre{t'} = \emptyset$ hence once that
 a transition is executed the places in its preset cannot be marked again.
\end{proof}
\begin{proposition}\label{pr:ca-unf}
 Let $\caname = \langle S, T, F, I, R, \mathsf{m}, \ell\rangle$ be a causal net. Then
 $\caname$ is an unfolding.
\end{proposition}
\begin{proof}
 The transitions bearing the same label are in conflict, hence the thesis.
\end{proof}

\aggiunta{In the following we show that classical notions like occurrence nets or unravel nets
can be seen as causal ones, showing that this notion is adequate to be related
to prime or bundle event structures. In the next section we will see that it is also
suitable to be related to \Ges{} as well.}   
\subsection{Occurrence nets and causal nets} 
\corr{To give further evidence that this notion could be the appropriate one,
w}{W}e show that each occurrence net can be turned into a causal one, thus this is a conservative extension 
of this notion.
The idea behind the construction is simple: to each event of the occurrence net a
transition in the causal net is associated, the places in the preset
of all transitions are initially marked and they are not in the postset of any other transition. 
The dependencies between events are modeled using inhibitor arcs. All the conflicts are modeled 
like in a conflict saturated net (with suitable marked places).
\begin{proposition}\label{pr:occtocausal}
 Let $\cnname = \langle B, E, F, \mathsf{c}\rangle$ be an occurrence net. The net
 $\ontocn{\cn}(\cnname)$ defined as $\langle S, E, F', I, \emptyset, \mathsf{m}, \ell\rangle$
 where
  \begin{itemize}
    \item $S = \setcomp{(\ast,e)}{e\in E}\cup\setcomp{(e,\ast)}{e\in E}
           \cup\setcomp{\setenum{e,e'}}{e\ \#\ e'}$,
    \item $F' = \setcomp{(s,e)}{s = (\ast,e)}\cup \setcomp{(s,e)}{e\in s} \cup 
          \setcomp{(e,s)}{s = (e,\ast)}$,
    \item $I = \setcomp{(s,e)}{s = (\ast,e')\ \land\ e' <_{\cnname} e}$, 
%    \item $R = \emptyset$, 
    \item $\mathsf{m} : S \to \nat$ is such that $\mathsf{m}(s) = 0$ if $s = (e,\ast)$ and
          $\mathsf{m}(s) = 1$ otherwise, and
    \item $\ell$ is the identity,       
  \end{itemize}
  is a causal net over the set of label $E$, and $\Conf{\cnname}{\cn} = \Conf{\ontocn{\cn}(\cnname)}{\ca}$.
\end{proposition}
\begin{proof}
 \aggiunta{First we recall that the labeling mapping is the identity, 
 therefore states and configurations coincide.~}
 An easy inspection show that the conditions of Definition~\ref{de:pre-causal-net} are 
 fulfilled by construction. Consider $X\in \states{\ontocn{\cn}(\cnname)}$ and assume
 that $\prec^{\ast}_{\ontocn{\cn}(\cnname)}$ is not a partial order. Then we have that
 for some $e, e'\in \flt{X}$, with $e\neq e'$ we have that 
 $e \prec^{\ast}_{\ontocn{\cn}(\cnname)} e'$ and $e' \prec^{\ast}_{\ontocn{\cn}(\cnname)} e$.
 Without loss of generality we can assume that 
 $e \prec_{\ontocn{\cn}(\cnname)} e'$ and $e' \prec_{\ontocn{\cn}(\cnname)} e$, but as 
 $\ontocn{\cn}(\cnname)$ is a pre-causal net by construction, this is impossible. 
 $\ontocn{\cn}(\cnname)$ satisfies also the other requirement of Definition~\ref{de:causal-net}, as
 clearly each transition can be executed.
 \aggiunta{Assume that there is a transition $e$ in $\ontocn{\cn}(\cnname)$ which is never executed,
 but it is in $\cnname$, as the latter is an occurrence net. As $e$ is executable in 
 $\cnname$ there is a firing sequence $\sigma = \sigma'\trans{e}m$ and the trace associated to this 
 \fs{}  $\trace{\sigma}$ is such that $\toset{\trace{\sigma}} = \hist{e}$. 
 Now assume that for each $i < \len{\trace{\sigma}}$ it holds that $\trace{\sigma}(i) = e_i$ 
 can be executed also in $\ontocn{\cn}(\cnname)$, and $\inib{e}$ in $\ontocn{\cn}(\cnname)$
 contains precisely the places $\setcomp{(e_i,\ast)}{e_i\in \toset{\trace{\sigma}}}$, which implies
 that also $e$ is executable in $\ontocn{\cn}(\cnname)$.~}

 To show that $\Conf{\cnname}{\cn} = \Conf{\ontocn{\cn}(\cnname)}{\ca}$ we prove
 that if $X\in \Conf{\cnname}{\cn}$ then $\mathsf{m} - \pre{X} + \post{X}$ is a reachable
 marking in $\ontocn{\cn}(\cnname)$. The proof is done by induction of the
 size of $X$. If $X$ is the emptyset then we are done. Assume it holds for
 $\len{X} = n$, hence $m' = \mathsf{m} - \pre{X} + \post{X}$ is a reachable
 marking in $\ontocn{\cn}(\cnname)$. Consider now the configuration $X\cup\setenum{e}$ for some
 $e\in E$. As it is a configuration we know that $\hist{e} \subseteq X$, hence for all
 $e'\in \hist{e}$ we have that  $m'((\ast,e)) = 0$, hence $m'\trans{e}(m-\pre{e}\cup\post{e})$, which means
 that to $X\cup\setenum{e}\in \Conf{\cnname}{\cn}$ a reachable marking in 
 $\ontocn{\cn}(\cnname)$ corresponds, and
 the thesis holds. The vice-versa, namely that 
 if $X \in  \Conf{\ontocn{\cn}(\cnname)}{\ca}$ then also $X\in\Conf{\cnname}{\cn}$, is proved analogously.
\end{proof}

Below we depict a simple occurrence net (on the left) and the associated causal one.
\aggiunta{The causal net on the right is conflict saturated.}
\medskip

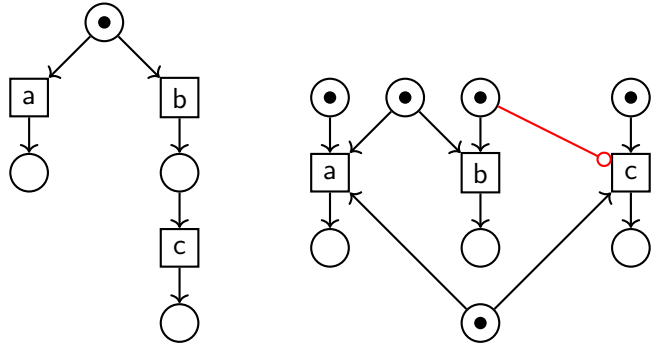
\begin{center}
\begin{tikzpicture}
\tikzstyle{inhibitor}=[o-,thick, draw=red]
\tikzstyle{pre}=[<-,thick]
\tikzstyle{post}=[->,thick]
\tikzstyle{read}=[-,thick]
\tikzstyle{transition}=[rectangle, draw=black,thick,minimum size=5mm]
\tikzstyle{place}=[circle, draw=black,thick,minimum size=5mm]
\node[place,tokens=1] (p1) at (5,2) {};
\node[place,tokens=1] (p3) at (7,2) {};
\node[place,tokens=1] (p5) at (9,2) {};
\node[place,tokens=1] (p7) at (6,2) {};
\node[place] (p2) at (5,0) {};
\node[place] (p4) at (7,0) {};
\node[place] (p6) at (9,0) {};
\node[place,tokens=1] (p8) at (7,-1) {};
\node[transition] (t1) at (5,1) {$\pmv{a}$}
edge[pre] (p1)
edge[pre] (p7)
edge[pre] (p8)
edge[post] (p2);
\node[transition] (t2) at (7,1) {$\pmv{b}$}
edge[pre] (p3)
edge[pre] (p7)
edge[post] (p4);
\node[transition] (t3) at (9,1) {$\pmv{c}$}
edge[pre] (p5)
edge[pre] (p8)
edge[inhibitor] (p3)
edge[post] (p6);
\node[place,tokens=1] (q1) at (2,3) {};
\node[place] (q2) at (1,1) {};
\node[place] (q3) at (3,1) {};
\node[place] (q4) at (3,-1) {};
\node[transition] (tt1) at (1,2)  {$\pmv{a}$}
edge[pre] (q1)
edge[post] (q2);
\node[transition] (tt2) at (3,2)  {$\pmv{b}$}
edge[pre] (q1)
edge[post] (q3);
\node[transition] (tt3) at (3,0)  {$\pmv{c}$}
edge[pre] (q3)
edge[post] (q4);
\end{tikzpicture}
\end{center}
\medskip

\aggiunta{Dependencies in a causal net corresponding to an occurrence net are here represented with
inhibitor arcs, but an alternative choice would have been to use read arcs instead. Thus
if $e <_{\cnname} e'$ in the occurrence net, 
instead of stipulating the existence of an inhibitor arc connecting
the place $(\ast,e)$ to the transition $e$, one can introduce a read arc from
the place $(e,\ast)$ to the transition $e$.~}

The obvious consequence of the previous Proposition is the following one.
\begin{proposition}\label{pr:occtocausal-conflictsat}
 Let $\cnname$ be an occurrence net and 
 $\ontocn{\cn}(\cnname)$ be the associated causal net. Then $\ontocn{\cn}(\cnname)$ is an unfolding.
\end{proposition}

We remark that in a causal net, 
it is a bit more tiring to dig out dependencies with respect to what 
happens in an occurrence net. 

The causal nets corresponding to occurrence nets have some further characteristic.
In fact the dependency relation $\prec$ 
and the conflict relation have a particular shape.

\begin{definition}\label{de:causal-occnet}
 Let $\caname = \langle S, T, F, I, R, \mathsf{m}, \ell\rangle$ be a causal net. 
 $\caname$ is said to be an \emph{occurrence causal} net whenever $R = \emptyset$, $\ell$ is injective,
 $\prec^{\ast}_{\caname}$ is a partial
 order over $T$, and if $t\ \#_{\caname}\ t' \prec^{\ast}_{\caname} t''$ then $t\ \#_{\caname}\ t''$.
\end{definition}
The above definition simply guarantees that the dependencies give a partial order and that the conflict
relation is inherited along the reflexive and transitive closure of the dependency relation.

\begin{proposition}
 Let $\cnname$ be an occurrence net and $\ontocn{\cn}(\cnname)$ be the associated causal net.
 Then $\ontocn{\cn}(\cnname)$ is an occurrence causal net.
\end{proposition}
\begin{proof}
 By construction $R = \emptyset$ and $\ell$ is the identity, hence it it injective. 
 $\prec^{\ast}_{\ontocn{\cn}(\cnname)}$ is a partial order as 
 we have that $e\prec _{\ontocn{\cn}(\cnname)} e'$ when $(\ast,e)\in\inib{e'}$ and
 this happen if $e < e'$ in the occurrence net $\cnname$, thus $\prec^{\ast}_{\ontocn{\cn}(\cnname)}$
 is a partial order as $<$ is a partial order. 
 As $\ontocn{\cn}(\cnname)$ is conflict saturated, it is clear that conflicts are inherited 
 along the $\prec^{\ast}_{\ontocn{\cn}(\cnname)}$ relation.
\end{proof}

\subsection{Causal nets and prime event structure}
Proposition~\ref{pr:occtocausal}, together with the connection among
\pes\ and \cn\ (Definition~\ref{de:pes-to-cn} and Proposition~\ref{pr:pestoon-classic}), 
suggests that a relation between \pes\ and \ca\ can be established.
Here the intuition is to use the same construction hinted in Proposition~\ref{pr:occtocausal}.
\begin{definition}\label{de:pestoca}
 Let $\pesname = (E, <, \#)$ be a \pes. Define 
 $\estonet{\pes}{\ca}(\pesname)$ as the labeled contextual Petri net 
 $\langle S, E, F, I, \emptyset, \mathsf{m}, \ell\rangle$ where
 \begin{itemize}
   \item $S = \setcomp{(\ast,e)}{e\in E}\cup 
              \setcomp{(e,\ast)}{e\in E}\cup
              \setcomp{(\setenum{e,e'},\#)}{e\ \#\ e'}$,
   \item $F =   \setcomp{(e,s)}{s = (e,\ast)}\ \cup\ 
                \setcomp{(s,e)}{s = (\ast,e) \lor\ (s = (W,\#)\ \land\ e\in W)}$,    
   \item $I = \setcomp{(s,e)}{s = (\ast,e')\ \land\ e' < e}$,       
   \item $\mathsf{m} = \setcomp{(\ast,e)}{e\in E}\cup 
                       \setcomp{(\setenum{e,e'},\#)}{e\ \#\ e'}$, and
   \item $\ell$ is the identity.
 \end{itemize}
\end{definition} 

\begin{proposition}\label{pr:pestocaandconf}
  Let $\pesname$ be a \pes, and $\estonet{\pes}{\ca}(\pesname)$ be the associated Petri net.
  Then $\estonet{\pes}{\ca}(\pesname)$ is a causal net and 
  $\Conf{\pesname}{\pes} = \Conf{\estonet{\pes}{\ca}(\pesname)}{\ca}$.
\end{proposition}
\begin{proof}
  Trivial, by observing that $\estonet{\pes}{\ca}(\pesname)$ is the same construction of
  Proposition~\ref{pr:occtocausal}.
\end{proof}
The construction gives an occurrence causal net.
\begin{proposition}\label{pr:pestoca}
  Let $\pesname$ be a \pes, and $\estonet{\pes}{\ca}(\pesname)$ be the associated Petri net.
  Then $\estonet{\pes}{\ca}(\pesname)$ is an occurrence causal net.
\end{proposition}
\begin{proof}
 Easy inspection of the construction in Definition~\ref{de:pestoca}.
\end{proof}
 
We show that also the vice versa is feasible provided that we restrict our attention to 
occurrence causal net.

\begin{proposition}\label{pr:caon-to-pes}
 Let $\caname = \langle S, T, F, I, R, \mathsf{m}, \ell\rangle$ be an occurrence causal net. Then
 $\nettoes{\ca}{\pes}(\caname) = (T, \prec^{+}_{\caname}, \#_{\caname})$
 is a \pes, and 
 $\Conf{\caname}{\ca} = \Conf{\nettoes{\ca}{\pes}(\caname)}{\pes}$.
\end{proposition}
\begin{proof}
 We have just to show that $\Conf{\caname}{\ca} = \Conf{\nettoes{\ca}{\pes}(\caname)}{\pes}$,
 as $\prec^{+}_{\caname}$ is a partial order being $\caname$ an occurrence causal net and
 it is inherited along the conflict relation $\#_{\caname}$. 
 Take $X\in \Conf{\caname}{\ca}$, then we have that 
 $X$ is conflict free and for each $t\in X$ it holds that all the transitions
 $t' \prec_{\caname} t$ have been executed as $t' = \post{s}$ for $s\in \inib{t}$, but this 
 account to say that $\hist{t}\subseteq X$ and this prove that $X \in 
 \Conf{\nettoes{\ca}{\pes}(\caname)}{\pes}$ as well.
 For the other \corr{containment}{inclusion}, consider $X \in \Conf{\nettoes{\ca}{\pes}(\caname)}{\pes}$.
 $X$ can be seen as the sequence $t_1t_2t_3\dots$ compatible with $\prec_{\caname}$.
 Take $t_n\in X$, we show that $m_n\trans{t_n}$ where 
 $m_n = \mathsf{m} - \pre{\setenum{T_1, \dots, t_{n-1}}} + \post{\setenum{T_1, \dots, t_{n-1}}}$.
 For all $t\in \hist{t_n}$ we have that $m_n(s) = 0$ for $s\in\pre{t}$ and
 as $t'\in \hist{t_n}$ we have that there is an inhibitor arc from an $s\in\pre{t'}$ and
 $t$, but then $m_n\trans{t_n}$, and the thesis follows.
\end{proof}

An important consequence of this proposition is that we can associate an occurrence net to
an occurrence causal one, simply by associating to the occurrence causal net the corresponding \pes\ as
in Proposition~\ref{pr:caon-to-pes} and then apply the construction in Definition~\ref{de:pes-to-cn}.

\begin{proposition}\label{pr:ca-to-on}
 Let $\caname = \langle S, T, F, I, R, \mathsf{m}, \ell\rangle$ be an occurrence causal net, then
 $\tuple{B, T, F', \mathsf{c}}$ where 
 \begin{itemize}
 \item $B = \setcomp{(\ast,e)}{e\in E}\cup \setcomp{(e,\ast)}{e\in E}\cup
         \setcomp{(e,e',\prec_{\caname})}{e \prec_{\caname} e'} \cup 
         \setcomp{(\setenum{e,e'},\#)}{e\ \#_{\caname}\ e'}$,
   \item $F =   \setcomp{(e,b)}{b = (e,\ast)}\ \cup\ \setcomp{(e,b)}{b = (e,e',\prec_{\caname})}\ \cup\ 
                \setcomp{(b,e)}{b = (\ast,e)}\ \cup\ \setcomp{(b,e)}{b = (e',e,\prec_{\caname})}\ \cup\ 
               \setcomp{(b,e)}{b = (Z,\#_{\caname})\ \land\ e\in Z}$,    
         and
   \item $\mathsf{c} = \setcomp{(\ast,e)}{e\in E}\cup \setcomp{(\setenum{e,e'},\#)}{e\ \#\ e'}$
  \end{itemize}   
  is an occurrence net.                                              
\end{proposition}
\begin{proof}
 By Propositions~\ref{pr:caon-to-pes} and \ref{pr:pestoon-classic}.
\end{proof}

The following theorem indeed assures that the notion of (occurrence) causal net is adequate as
the notion of occurrence net with respect to the classical notion of occurrence net in the relationship 
with prime event structure. 

\begin{theorem}
 Let $\caname$ be an occurrence causal net and  $\pesname$ be a prime event structure.
 Then $\caname \cong \estonet{\pes}{\ca}(\nettoes{\ca}{\pes}(\caname))$ and 
 $\pesname \equiv \nettoes{\ca}{\pes}(\estonet{\pes}{\ca}(\pesname))$.
\end{theorem}
\begin{proof}
 By Propositions~\ref{pr:pestocaandconf} and~\ref{pr:caon-to-pes}.
\end{proof}

%\michelenote{mettere la parte sulla costruzione del causal net da un unravel net, la prova che tutto funziona
%dipende dal fatto che ogni stato da una cosa aciclica}

\subsection{Unravel nets and causal nets}
In an unravel net the same transitions may have various causes, conflicting one. In a causal
net we may have various incarnations of the same \emph{activity}, depending on different contexts. 
In associating a causal net to an unravel one $\unname$ we have to establish, for each transition in
$\unname$, the possible sets of causes.
%\michelenote{la devo definire sulle etichette}.
\begin{definition}\label{de:un-net-causes}
  Let $\unname = \tuple{S, T, F, \marko{m}}$ be an \un, and let $t\in T$ be a transition.
  Then the set $\causes{t}$ defined as 
  $\setcomp{Y\subseteq T}{\forall s\in\pre{t}. \pre{s}\neq\emptyset\ \Rightarrow\ 
  |\pre{s}\cap Y| = 1}\cup\setcomp{\emptyset}{\forall s\in\pre{t}. \pre{s} = \emptyset}$ is the set of all possible dependencies for $t$.
\end{definition}
We add the $\emptyset$ when $t$ does not depend on any other occurrence of a transition.

We show how to associate a causal net to an unravel one. 
The idea is to create a copy of the transition $t$ for each of its possible dependencies in the
unravel net, and each of these dependencies are then implemented with inhibitor arcs.
\begin{proposition}\label{pr:un-to-ca}
  Let $\unname = \tuple{S, T, F, \marko{m}}$ be an \un.
  The net $\ontocn{\un}(\unname) = \langle S, T', F', I, \emptyset, \hat{\mathsf{m}}, \ell\rangle$
  where
  \begin{itemize}
    \item $S = \setcomp{(\ast,t)}{t\in T}\cup\setcomp{(t,\ast)}{t\in T}
           \cup\setcomp{\setenum{t,t'}}{t\ \#\ t'}$,
    \item $T' = \setcomp{(t,Y)}{t\in T\ \land\ Y\in \causes{t}}$,
    \item $F' = \setcomp{(s,(t,-))}{s = (\ast,t)}\cup \setcomp{(s,(t,-))}{t\in s} \cup 
          \setcomp{((t,-),s)}{s = (t,\ast)}$,
    \item $I = \setcomp{(s,(t,Y))}{s = (\ast,t')\ \land\ t'\in Y}$, 
%    \item $R = \emptyset$, 
    \item $\hat{\mathsf{m}} : S \to \nat$ is such that $\hat{\mathsf{m}}(s) = 0$ if $s = (t,\ast)$ and
          $\hat{\mathsf{m}}(s) = 1$ otherwise, and
    \item $\ell(t,-) = t$ is the labeling mapping,       
  \end{itemize}
  is a causal net over the set of label $T$, and $\Conf{\unname}{\un} = \Conf{\ontocn{\un}(\unname)}{\ca}$.
\end{proposition}
\begin{proof}
 We check the various conditions of Definitions~\ref{de:pre-causal-net} and~\ref{de:causal-net}.
 We start with the ones of Definition~\ref{de:pre-causal-net}. Conditions (1), (3) and (5) are
 trivially satisfied by construction. For the second condition, take any
 transition $(t,Y)\in T'$, with $t\in T$ and $Y\subseteq T$ and $Y\neq\emptyset$, 
 and  consider $s\in \inib{(t,Y)}$. We have that 
 $s$ is of the form $(\ast,t')$ for some $t'\in Y$. Clearly 
 $\post{s} = \setenum{(t',Y'_1), \dots (t',Y'_k)}$ and $\ell(\post{s}) = \setenum{t'}$, thus
 $|\ell(\post{s})| = 1$. 
 Given a $t\in T'$, finiteness of $\inib{t}$ derives by the fact that in the \un\ $\unname$ there
 is no place with infinite incoming arcs, and the last condition follows by the fact that
 equally labeled transitions $t$ and $t'$ in $T'$ share a common input place: $(\ast,\ell(t))$.
 
 To show that $\Conf{\unname}{\un} = \Conf{\ontocn{\un}(\unname)}{\ca}$ we proceed as follows:
 consider $\sigma$ a \fs\ in $\unname$ and $X$ the corresponding state. We show that 
 there is a \fs\ $\widehat{\sigma}$ in $\ontocn{\un}(\unname)$ such that 
 $\trace{\sigma} = \trace{\widehat{\sigma}}$. 
 In particular we prove that if $\sigma = \sigma'\trans{t}\sigma''$ then
 there is a transition $(t,Y)$ such that 
 $\widehat{\sigma} = \widehat{\sigma'}\trans{(t,Y)}\widehat{\sigma''}$ with 
 $\trace{\sigma'} = \trace{\widehat{\sigma'}}$. By induction on the length of the firing
 sequence $\sigma'$, we have that if $\lung{\sigma'} = 0$ then $\sigma' = \marko{m}$ and
 consider $t$ enabled at $\marko{m}$. Then $\widehat{\sigma'} = \hat{\marko{m}}$ and
 $(t,\emptyset)$ is clearly enabled at this marking. 
 Consider then $\sigma'$, $\toset{\trace{\sigma'}}$ and a transition $t$ enabled at
 $\lead{\sigma'}$, we have that for exactly one $Y\subseteq \toset{\trace{\sigma'}}$ it
 holds that $(t,Y)$ is enabled at the marking $\lead{\widehat{\sigma'}}$ as the
 transitions in $Y$ are those that have marked the preset of $t$ in $\unname$. 
 A similar argument show that to each firing sequence $\widehat{\sigma}$ in $\ontocn{\un}(\unname)$ a 
 firing sequence $\sigma$ in $\unname$ corresponds and $\trace{\sigma} = \trace{\widehat{\sigma}}$.
 Thus $\Conf{\unname}{\un} = \Conf{\ontocn{\un}(\unname)}{\ca}$.
\end{proof}

\begin{example}\label{ex:unravascn}
Below we depict a simple unravel net (on the left) and the associated causal one.
\medskip

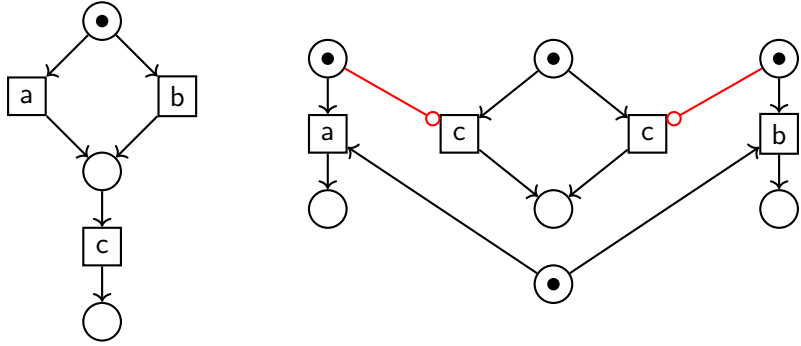
\begin{center}
\begin{tikzpicture}
\tikzstyle{inhibitor}=[o-,thick,draw=red]
\tikzstyle{pre}=[<-,thick]
\tikzstyle{post}=[->,thick]
\tikzstyle{read}=[-,thick]
\tikzstyle{transition}=[rectangle, draw=black,thick,minimum size=5mm]
\tikzstyle{place}=[circle, draw=black,thick,minimum size=5mm]
\node[place,tokens=1] (p1) at (5,2.5) {};
\node[place,tokens=1] (p3) at (8,2.5) {};
\node[place,tokens=1] (p5) at (11,2.5) {};
\node[place] (p2) at (5,.5) {};
\node[place] (p4) at (8,.5) {};
\node[place] (p6) at (11,.5) {};
\node[place,tokens=1] (p8) at (8,-.5) {};
\node[transition] (t1) at (5,1.5) {$\pmv{a}$}
edge[pre] (p1)
edge[pre] (p8)
edge[post] (p2);
\node[transition] (t2) at (6.75,1.5) {$\pmv{c}$}
edge[pre] (p3)
edge[inhibitor] (p1)
edge[post] (p4);
\node[transition] (t4) at (9.25,1.5) {$\pmv{c}$}
edge[pre] (p3)
edge[inhibitor] (p5)
edge[post] (p4);
\node[transition] (t3) at (11,1.5) {$\pmv{b}$}
edge[pre] (p5)
edge[pre] (p8)
edge[post] (p6);
\node[place,tokens=1] (q1) at (2,3) {};
\node[place] (q3) at (2,1) {};
\node[place] (q4) at (2,-1) {};
\node[transition] (tt1) at (1,2)  {$\pmv{a}$}
edge[pre] (q1)
edge[post] (q3);
\node[transition] (tt2) at (3,2)  {$\pmv{b}$}
edge[pre] (q1)
edge[post] (q3);
\node[transition] (tt3) at (2,0)  {$\pmv{c}$}
edge[pre] (q3)
edge[post] (q4);
\end{tikzpicture}
\end{center}

Observe that there are two incarnations for the same transition $\pmv{c}$ in the corresponding \ca\ one,
namely $(\pmv{c},\setenum{\pmv{a}})$ for the one on the left and $(\pmv{c},\setenum{\pmv{b}})$
for the one on the right.
\end{example}
On a \ca\ net we can define a conflict relation involving the labels. 
Given the \ca\ $\caname = \langle S, T, F, I, R, \mathsf{m}, \ell\rangle$ over the set of labels
$\Lab$, we say that $a\ \sharp\ b$ whenever $\forall t\in \ell^{-1}(a)\  \forall t'\in \ell^{-1}(b).\ 
t\ \#\ t'$: in this conflict relation, defined on labels, we represent the fact that all the
possible incarnations of the activities $a$ and $b$ are in conflict.

The \ca\ obtained from an unravel net has a peculiar shape when looking at the contextual 
arcs representing the dependencies: inhibitor arcs connecting places and transitions are such that
the places have no incoming arcs. Furthermore two equally labeled transitions $t$ and $t'$ 
are \emph{inhibited} by the same number of \emph{transitions} and the labels of these 
inhibiting transitions are either in conflict or they are the same.
\begin{definition}\label{de:ca-plain-caus}
 Let $\caname = \langle S, T, F, I, R, \mathsf{m}, \ell\rangle$ be a \ca\ over the set of labels
 $\Lab$, we say that
 $\caname$ is \emph{plainly caused} whenever
 \begin{itemize}
  \item  $R = \emptyset$, 
  \item $\forall t\in T\ \forall s\in\inib{t}.\ \pre{s} = \emptyset$, 
  \item $\forall a\in \ell(T).\ \forall t,t\in\ell^{-1}(a).\ |\inib{t}| = |\inib{t'}|$, and
  \item $\forall t, t'\in T.\ (\ell(t) = \ell(t')\ \land\ t\neq t')\ \Rightarrow\ 
        \forall s\in\inib{t}\ \forall s'\in\inib{t}.\ (\ell(\post{s}) = \ell(\post{s'})\ 
        \lor\ \ell(\post{s})\ \sharp\ \ell(\post{s'}))$.
  \end{itemize}
\end{definition}

The \ca{s} obtained either by an \cn\ or by an \un\ are plainly caused (we have already observed that
an occurrence net is an unravel net as well).
\begin{proposition}
 Let $\unname $ be an \un\ and $\ontocn{\un}(\unname)$ be the associated \ca. Then 
 $\ontocn{\un}(\unname)$ is plainly caused.
\end{proposition}
\begin{proof}
 The first three conditions are clear by the construction in Proposition~\ref{pr:un-to-ca}, and
 for the last one it is enough to observe how $\causes{t}$ is formed.
\end{proof}

Given a label $a\in \Lab$ and a \ca\ $\caname = \langle S, T, F, I, R, \mathsf{m}, \ell\rangle$,
the set of the possible dependencies of the transitions $t\in \ell^{-1}(a)$ is defined 
as $\Ini{a} = \setcomp{\ell(\post{s})}{s\in \inib{t}\ \land\ t\in \ell^{-1}(a)}$\corr{ and}{.~Given 
an event $a$, this 
set contains all the labels (events) that should have happened in order to allow to one of the transitions
labeled with $a$ to occur. From this set } 
it is possible to extract maximal subsets of conflicting labels, which we denote
with $\maxbund{a} = \setcomp{Y\subseteq \Ini{a}}{\forall b,b'\in Y.\ b\ \sharp\ b'\ \land\ Y\ 
\mbox{ is maximal}}$.\aggiunta{~Maximality is required as otherwise some events would not occur in
the associated \bes{} as the required \emph{dependency} has been left out.~} 

\aggiunta{\begin{example}
Consider the $\caname$ net in Example~\ref{ex:unravascn}. Then $\Ini{\pmv{a}} = \emptyset$, 
$\Ini{\pmv{b}} = \emptyset$ 
and $\Ini{\pmv{c}} = \setenum{\pmv{a},\pmv{b}}$. 
The components of the bundles are determined using $\maxbund{\cdot}$ and in this case we have
$\maxbund{\pmv{c}} = \setenum{\setenum{\pmv{a},\pmv{b}}}$.
Observe that maximality is required as otherwise one could infer, \eg, $\setenum{\pmv{a}}\bundle \pmv{c}$,
which is not correct.
\end{example}}

\begin{proposition}\label{pr:caps-to-bes}
 Let $\caname = \langle S, T, F, I, R, \mathsf{m}, \ell\rangle$ be a plainly caused \ca\ over the
 set of labels $E$. Then $\nettoes{\ca}{\bes}(\caname) = (E, \bundle, \sharp)$ where 
 $Y \bundle e$ whenever $Y\in \maxbund{e}$ is a \bes, and 
 $\Conf{\caname}{\ca} = \Conf{\nettoes{\ca}{\bes}(\caname)}{\bes}$.
\end{proposition}
\begin{proof}
 Along the same lines as Proposition~\ref{pr:caon-to-pes}.
 \aggiunta{$\nettoes{\ca}{\bes}(\caname)$ is a \bes{} as each $Y \in \maxbund{e}$ is made of conflicting 
 events, and the conflict relation is irreflexive and symmetric by construction. 
 To show that the configurations coincide, consider $X\in \Conf{\caname}{\ca}$, and assume
 that $X = X_{\sigma}$ for some $\sigma\in\firseq{\caname}{\mathsf{m}}$. 
 To prove that $X\in \Conf{\nettoes{\ca}{\bes}(\caname)}{\bes}$ as well, we take
 $\trace{\sigma}$. Clearly $\trace{\sigma}(1) = e_1$ is such that $\Ini{e_1} = \emptyset$ 
 as $\caname$ is plainly caused, hence $\setenum{e_1}\in  \Conf{\nettoes{\ca}{\bes}(\caname)}{\bes}$. 
 Assume now that $\sigma = \sigma'\trans{t}\sigma''$ and that $\toset{\trace{\sigma'}}\in 
 \Conf{\nettoes{\ca}{\bes}(\caname)}{\bes}$, we have to show that $Y\in \maxbund{\ell(t)}$
 is such that $Y\cap \toset{\trace{\sigma'}}\neq \emptyset$, but this is clearly true as
 $Y\in \maxbund{\ell(t)}$ and because the set of inhibition has the same number of elements
 for each incarnation of $\ell(t)$. Conflicting events cannot be added, hence the thesis. 
 The other inclusion holds with the same reasoning.}
\end{proof}

Given a \bes\ $\besname$ define $\estonet{\bes}{\ca}(\besname)$ as 
$\ontocn{\un}(\estonet{\bes}{\un}(\besname))$. Then also the following theorem holds.

\begin{theorem}
 Let $\caname$ be an plainly caused causal net and  $\besname$ be a bundle event structure.
 Then $\caname \cong \estonet{\bes}{\ca}(\nettoes{\ca}{\bes}(\caname))$ and 
 $\besname \equiv \nettoes{\ca}{\bes}(\estonet{\bes}{\ca}(\besname))$.
\end{theorem}
\begin{proof}
 By Propositions~\ref{pr:un-to-ca} and~\ref{pr:caps-to-bes}.
\end{proof}

\section{Context-dependent event structures and causal nets}\label{sec:cdes-causal}
We are now ready to relate Context-dependent event structures and causal nets.
We recall that in a Context-dependent event structure each event may happen in
different context and thus each happening has a different operational meaning.
Therefore we model each happening with a different transition and all the transitions 
representing the same happening bear the same label.
Dependencies are inferred using inhibitor and read arcs, as it will be clear.

\begin{definition}\label{de:cdestoca}
 Let $\cdesname = (E, \#, \gesrel)$ be an elementary \Ges\ such that
 $\forall \pmv{Z}\gesrel e.\ \ctx(\pmv{Z}\gesrel e)$ is finite. 
 Define 
 $\estonet{\Ges}{\ca}(\cdesname)$ as the net 
 $\langle S, T, F, I, R, \mathsf{m}, \ell\rangle$ where
% \begin{itemize}
%   \item $\begin{array}{lcl}
%          S & = & \setcomp{(\ast,e)}{e\in E}\cup 
%              \setcomp{(e,\ast)}{e\in E}\cup
%              \setcomp{(\setenum{e,e'},\#)}{e\ \#\ e'} \\
%         \end{array}$,
%   \item $\begin{array}{lcl}
%          T & = & \setcomp{(e,X,Y)}{(X,Y)\in\pmv{Z}\ \land\ \pmv{Z}\gesrel e}
%          \end{array}$,      
%   \item $\begin{array}{lcl}
%           F & = & \setcomp{(s,(e,X,Y))}{s = (\ast,e)\ \lor\ 
%              (s = (W,\#)\ \land\ e\in W)}\ \cup\ \\
%              \end{array}$
%              
%           $\begin{array}{ll}   
%           \hspace*{.65cm} & \setcomp{((e,X,Y),s)}{s = (e,\ast)} \\
%          \end{array}$,    
%   \item $\begin{array}{lcl}
%            I & = & \setcomp{(s,(e,X,Y))}{s = (e',\ast)\ \land\ e' \in 
%                   \ctx(\pmv{Z}\gesrel e)\setminus (X\cup Y)}\ 
%              \cup\ \\
%               \end{array}$
%              
%           $\begin{array}{ll}   
%           \hspace*{.56cm} &  \setcomp{(s,(e,X,Y))}{s = (\ast,e')\ \land\ e' \in X}\\
%         \end{array}$, 
%   \item $\begin{array}{lcl}
%           R & = & \setcomp{(s,(e,X,Y))}{s = (e',\ast)\ \land\ e' \in Y} \\
%           \end{array}$,      
%   \item $\begin{array}{lcl}
%         \mathsf{m} & = & \setcomp{(\ast,e)}{e\in E}\cup \setcomp{(\setenum{e,e'},\#)}{e\ \#\ e'} \\
%         \end{array}$, and
%   \item $\ell : T \to E$ is defined as $\ell((e,X,Y)) = e$.
% \end{itemize}
\begin{itemize}
   \item $S = \setcomp{(\ast,e)}{e\in E}\cup 
              \setcomp{(e,\ast)}{e\in E}\cup
              \setcomp{(\setenum{e,e'},\#)}{e\ \#\ e'}$,
   \item $T = \setcomp{(e,X,Y)}{(X,Y)\in\pmv{Z}\ \land\ \pmv{Z}\gesrel e}$,      
   \item $F = \setcomp{(s,(e,X,Y))}{s = (\ast,e)\ \lor\ 
              (s = (W,\#)\ \land\ e\in W)}\ \cup\ \setcomp{((e,X,Y),s)}{s = (e,\ast)}$,    
   \item $I = \setcomp{(s,(e,X,Y))}{s = (\ast,e')\ \land\ e' \in X}\
              \cup$\\ 
              \hspace*{.7cm}$\setcomp{(s,(e,X,Y))}{s = (e',\ast)\ \land\ e' \in 
                   \ctx(\pmv{Z}\gesrel e)\setminus (X\cup Y)}$, 
   \item $R = \setcomp{(s,(e,X,Y))}{s = (e',\ast)\ \land\ e' \in Y}$,      
   \item $\mathsf{m}=\setcomp{(\ast,e)}{e\in E}\cup \setcomp{(\setenum{e,e'},\#)}{e\ \#\ e'}$, and
   \item $\ell : T \to E$ is defined as $\ell((e,X,Y)) = e$.
 \end{itemize}
\end{definition}
We introduce a transition $(e,X,Y)$ for each pair $(X,Y)$ of the entry associated
to the event $e$, and all these transitions are labeled with the same event $e$.
All these transitions consume the token present in the place $(\ast,e)$ and put 
a token in the place $(e,\ast)$, thus just one transition labeled with 
$e$ can be fired in each execution of the net. 
Recall that the event $e$ is enabled at a configuration $C$ (here signaled by the
places $(e',\ast)$ marked) if, for some
$(X,Y)\in\pmv{Z}$, it holds that $\ctx(\pmv{Z}\gesrel e)\cap C = X$ and $Y\subseteq C$.
The inhibitor arcs assure that some of the events in $\ctx(\pmv{Z}\gesrel e)$ 
have actually happened (namely the one in $X$) but the others 
(the ones in $\ctx(\pmv{Z}\gesrel e)\setminus (X\cup Y)$)
have not, and the $Y$ are other events that must have happened and this is signaled by read arcs.
We cannot require $\ctx(\pmv{Z}\gesrel e)\setminus X$ as some of the events there may be present
in $Y$. 

\begin{example}\label{ex:new2dim-cn}
 Consider the \Ges\ in Example~\ref{ex:new2dim}, the corresponding causal net is the one 
 depicted in Example~\ref{ex:causalnet}. The event $\pmv{c}$ has two incarnations as the entry
 $\setenum{(\emptyset,\setenum{\pmv{a}}),(\setenum{\pmv{b}},\emptyset)}\gesrel \pmv{c}$ 
 has two elements: $(\emptyset,\setenum{\pmv{a}})$ and $(\setenum{\pmv{b}},\emptyset)$.
\end{example} 
\begin{example}\label{ex:new2aug-cn} 
 Consider the \Ges\ of the Example~\ref{ex:new2aug},
 the event $\pmv{c}$ has two incarnations as the entry
 $\setenum{(\emptyset,\emptyset),(\setenum{\pmv{b}},\setenum{\pmv{a}})}\gesrel \pmv{c}$
 has two elements, whereas $\pmv{a}$ and $\pmv{b}$ have one. 
 The associated causal net is
 \begin{center}
  \scalebox{0.9}{\begin{tikzpicture}
\tikzstyle{inhibitor}=[o-,thick, draw=red]
\tikzstyle{pre}=[<-,thick]
\tikzstyle{post}=[->,thick]
\tikzstyle{read}=[-,thick]
\tikzstyle{readblu}=[-,draw=blue,thick]
\tikzstyle{transition}=[rectangle, draw=black,thick,minimum size=5mm]
\tikzstyle{place}=[circle, draw=black,thick,minimum size=5mm]
\node[place,tokens=1] (p1) at (0,2) {};
\node[place,tokens=1] (p3) at (3,1.5) {};
\node[place,tokens=1] (p5) at (6,2) {};
\node[place] (p2) at (0,0) {};
\node[place] (p4) at (3,0) {};
\node[place] (p6) at (6,0) {};
\node[transition] (t1) at (0,1) [label=left:$\pmv{b}$] {$t_1$}
edge[pre] (p1)
edge[post] (p2);
\node[transition] (t2) at (1.5,1) [label=left:$\pmv{c}$] {$t_2$}
edge[pre] (p3)
edge[inhibitor] (p2)
edge[post] (p4);
\node[transition] (t3) at (4.5,1) [label=right:$\pmv{c}$] {$t_3$}
edge[pre] (p3)
edge[readblu] (p6)
edge[inhibitor,bend right] (p1)
edge[post] (p4);
\node[transition] (t4) at (6,1) [label=right:$\pmv{a}$] {$t_4$}
edge[pre] (p5)
edge[post] (p6);
\end{tikzpicture}}
 \end{center}
\end{example} 

\begin{example}\label{ex:new2res-cn}
Consider now the \Ges\ in Example~\ref{ex:new2res} (modeling the resolvable conflict of \cite{GP:ESRC}).
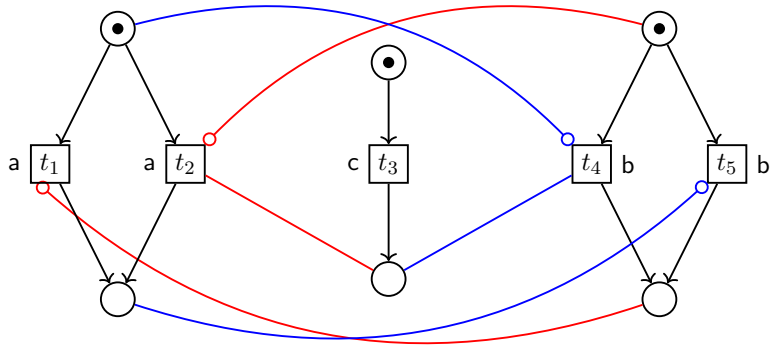
  \begin{center}
  \scalebox{0.9}{\begin{tikzpicture}
\tikzstyle{inhibitor}=[o-,thick]
\tikzstyle{inhibitorred}=[o-,draw=red,thick]
\tikzstyle{inhibitorblu}=[o-,draw=blue,thick]
\tikzstyle{inhibitorpur}=[o-,draw=purple,thick]
\tikzstyle{pre}=[<-,thick]
\tikzstyle{post}=[->,thick]
\tikzstyle{read}=[-,thick]
\tikzstyle{readred}=[-,draw=red,thick]
\tikzstyle{readblu}=[-,draw=blue,thick]
\tikzstyle{readpur}=[-,draw=purple,thick]
\tikzstyle{transition}=[rectangle, draw=black,thick,minimum size=5mm]
\tikzstyle{invtransition}=[rectangle, draw=black!0,thick,minimum size=6mm]
\tikzstyle{place}=[circle, draw=black,thick,minimum size=5mm]
\node[place,tokens=1] (p1) at (1,3.5) {};
\node[place,tokens=1] (p3) at (5,3) {};
\node[place,tokens=1] (p5) at (9,3.5) {};
\node[place] (p2) at (1,-0.5) {};
\node[place] (p4) at (5,-0.2) {};
\node[place] (p6) at (9,-0.5) {};
\node[invtransition] (t1ai) at (-0.5,1.5) {}
edge[inhibitorred, bend right] (p6);
\node[transition] (t1a) at (0,1.5) [label=left:$\pmv{a}$] {$t_1$}
edge[pre] (p1)
%edge[inhibitorred, bend right] (p4)
edge[post] (p2);
%\node[transition] (t2a) at (1,1.5) [label=left:$\pmv{a}$] {$t_2$}
%edge[pre] (p1)
%edge[inhibitorred, bend left] (p3)
%edge[inhibitorred, bend right] (p6)
%edge[post] (p2);
\node[transition] (t3a) at (2,1.5) [label=left:$\pmv{a}$] {$t_2$}
edge[pre] (p1)
edge[inhibitorred, bend left] (p5)
edge[readred] (p4)
edge[post] (p2);
%\node[transition] (t1c) at (4,1.5) [label=left:$\pmv{c}$] {$t_4$}
%edge[pre] (p3)
%edge[inhibitorblu, bend left] (p2)
%edge[inhibitorblu, bend right] (p6)
%edge[post] (p4);
\node[transition] (t2c) at (5,1.5) [label=left:$\pmv{c}$] {$t_3$}
edge[pre] (p3)
%edge[inhibitorblu, bend right] (p1)
%edge[inhibitorblu, bend right] (p6)
edge[post] (p4);
%\node[transition] (t3c) at (6,1.5) [label=right:$\pmv{c}$] {$t_6$}
%edge[pre] (p3)
%edge[inhibitorblu, bend left] (p2)
%edge[inhibitorblu, bend left] (p5)
%edge[post] (p4);
\node[transition] (t1b) at (8,1.5) [label=right:$\pmv{b}$] {$t_4$}
edge[pre] (p5)
edge[inhibitorblu, bend right] (p1)
edge[readblu] (p4)
edge[post] (p6);
%\node[transition] (t2b) at (9,1.5) [label=right:$\pmv{b}$] {$t_8$}
%edge[pre] (p5)
%edge[inhibitorpur, bend right] (p3)
%edge[inhibitorpur, bend left] (p2)
%edge[post] (p6);
\node[transition] (t3b) at (10,1.5) [label=right:$\pmv{b}$] {$t_5$}
edge[pre] (p5)
edge[inhibitorblu, bend left] (p2)
edge[post] (p6);
\end{tikzpicture}}
 \end{center}
 the actual 
 implementation of this \Ges\ into the causal net depicted before, where 
 \corr{each}{the} event\corr{ has three incarnations.}{s $\pmv{a}$ and $\pmv{b}$ have two incarnations each.}
 \aggiunta{The transition $t_1$ is the one with $(X,Y) = (\emptyset,\emptyset)$ whereas 
 $t_2$ is $(\pmv{a},\setenum{\pmv{b}},\setenum{\pmv{c}})$, $t_4$ is the transition 
 $(\pmv{b},\setenum{\pmv{a}},\setenum{\pmv{c}})$ and finally $t_5$ is $(\pmv{b},\emptyset,\emptyset)$.~}
 The inhibitor and read arcs are colored depending on event they are related to.
\end{example}
The net obtained from a \Ges\ using Definition~\ref{de:cdestoca} is indeed a causal net, 
and furthermore it is also conflict saturated.
\begin{proposition}\label{pr:gestocaandconfig}
  Let $\cdesname$ be an elementary \Ges, and $\estonet{\Ges}{\ca}(\cdesname)$ be the associated contextual Petri net.
  Then $\estonet{\Ges}{\ca}(\cdesname)$ is a causal net and 
  $\Conf{\cdesname}{\Ges} = \Conf{\estonet{\Ges}{\ca}(\cdesname)}{\ca}$.
\end{proposition}
\begin{proof}
 We proceed along the same line as in the proof of Proposition~\ref{pr:un-to-ca} and we
 check the various conditions of Definitions~\ref{de:pre-causal-net} and~\ref{de:causal-net}.
 Conditions (1), (3) and (5) of Definition~\ref{de:pre-causal-net} are
 again satisfied by construction. For the second condition, take any
 transition $t = (e,X,Y)\in T$, with $e\in E$ and $X, Y\subseteq E$. 
 Consider $s\in \inib{t}$. We have that 
 $s$ is of the form $(\ast,e')$ for some $e'\in e' \in X$. 
 Clearly 
 $\post{s} = \setenum{(e',X_1,Y'_1), \dots (t',Y'_k)}$ and $\ell(\post{s}) = \setenum{e'}$, thus
 $|\ell(\post{s})| = 1$. 
 Given a $t\in T$, finiteness of $\inib{t}$ derives by the fact that in the \Ges\ 
 $\pmv{Z}$ is finite.
 The last condition follows by the fact that
 equally labeled transition $t$ and $t'$ in $T$ share a common input place: $(\ast,\ell(t))$.
 
 $\Conf{\cdesname}{\Ges} = \Conf{\estonet{\Ges}{\ca}(\cdesname)}{\ca}$ depends on the fact that
 to each configuration $C\in  \Conf{\cdesname}{\Ges}$ a reachable marking $m_C$ in 
 the \ca\ $\estonet{\Ges}{\ca}(\cdesname)$ correspond, and vice versa.
 
 First observe that the execution of a transition labeled with $e$ is signaled by the fact 
 that the place $(e,\ast)$ is marked (and $(\ast,e)$ not), and all the places
 $(\setenum{e.e'},\#)$ are unmarked, thus given a configuration $X$ the corresponding marking 
 $m_C$ is such that $m((e,\ast)) = 1$ if $e\in C$, $m((e,\ast)) = 1$ if $e\not\in C$,
 $m((W,\#)) = 1$ if $W\cap C = \emptyset$ and $m((W,\#)) = 0$ if $W\cap C \neq \emptyset$. 
 
 If $C$ is the empty configuration then clearly the initial marking does the job in 
 $\estonet{\Ges}{\ca}(\cdesname)$. Consider a configuration $C$ of size $n$, and
 the corresponding marking $m_C$, assume that $C\enab{e}$ for some $e\in E$.
 Recall that this means that in $\pmv{Z}\gesrel e$ there is a pair $(X,Y)$ such
 that $(\ctx{\pmv{Z}\gesrel e})\cap C = X$ and $Y\subseteq C$.
 Take the transition $(e,X,Y)$, it is clear that $m_C\trans{(e,X,Y)}$, as the inhibitor and read
 arcs guarantee exactly this. The vice versa is along the same line, and the  equality between
 $\Conf{\cdesname}{\Ges}$ and $\Conf{\estonet{\Ges}{\ca}(\cdesname)}{\ca}$ follows.
\end{proof}

\begin{proposition}\label{pr:gestoca-conflsat}
  Let $\cdesname$ be a \Ges, and $\estonet{\Ges}{\ca}(\cdesname)$ be the associated contextual Petri net.
  Then $\estonet{\Ges}{\ca}(\cdesname)$ is conflict saturated.
\end{proposition}
\begin{proof}
 Obvious as $\estonet{\Ges}{\ca}(\cdesname)$ is a \ca.
\end{proof}

For the vice versa we do need to make a further assumption on the causal net.
The intuition is that equally labeled transitions are different incarnations of the 
same activity, happening in different contexts. 
Henceforth one has to make sure that the equally labeled transitions indeed represent
the same \emph{event} and each incarnation of an event should have the same \emph{environment}, meaning
with environment the events related to it (which in the \Ges\ is calculated with $\ctx$).
Given a causal net $\caname = \langle S, T, F, I, R, \mathsf{m}, \ell\rangle$ on a set of label
$\Lab$ and a transition $t\in T$, with $\overrightarrow{\inib{t}}$ we denote the set of labels
$\setcomp{a\in \Lab}{s\in\inib{t}\ \land\ \ell(\post{s}) = a}$, with $\overleftarrow{\inib{t}}$ 
the set of labels
$\setcomp{a\in \Lab}{s\in\inib{t}\ \land\ \ell(\pre{s}) = a}$, and with 
$\widetilde{\rarc{t}}$ the set of labels $\setcomp{a\in \Lab}{s\in\rarc{t}\ \land\ \ell(\pre{s}) = a}$.

\begin{definition}\label{de:cawellbehaved}
 Let $\caname = \langle S, T, F, I, R, \mathsf{m}, \ell\rangle$ be a causal net labeled over
 $\Lab$, we say
 that $\caname$ is \emph{well behaved} if 
 \begin{enumerate}
   \item\label{c1:cnwellbehaved} $\forall a\in \Lab$. $\forall t, t'\in \ell^{-1}(a)$ it holds that  
         $\pre{t}\cup\pre{t'} = \setenum{s}$ and $\post{t}\cup\post{t'} = \setenum{s'}$, and
   \item\label{c2:cnwellbehaved} $\forall a\in \Lab$. $\forall t, t'\in \ell^{-1}(a)$ it holds that 
         $\overrightarrow{\inib{t}} \cup \widetilde{\rarc{t}} \cup \overleftarrow{\inib{t}} = 
          \overrightarrow{\inib{t'}} \cup \widetilde{\rarc{t'}} \cup \overleftarrow{\inib{t'}}$.
 \end{enumerate}   
\end{definition}
In a well behaved causal net all the transitions \corr{sharing}{} equally labeled 
have a common input place and also a common output place (condition~\ref{c1:cnwellbehaved}).
The equally labeled transitions in the causal net are the various incarnations of the
\emph{event} they represent, thus they have the same context, though the various kind of involved
arcs are different (condition~\ref{c2:cnwellbehaved}).
\begin{example}
 Consider the net below:
 
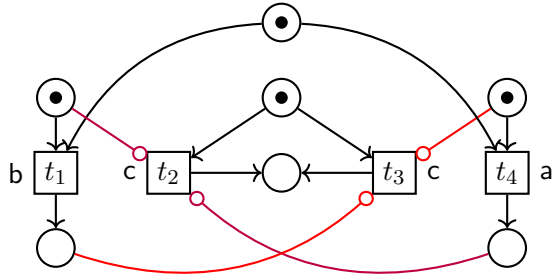
\begin{center}
  \begin{tikzpicture}
\tikzstyle{inhibitorred}=[o-,draw=red,thick]
\tikzstyle{inhibitorblu}=[o-,draw=purple,thick]
\tikzstyle{pre}=[<-,thick]
\tikzstyle{post}=[->,thick]
\tikzstyle{read}=[-,thick]
\tikzstyle{transition}=[rectangle, draw=black,thick,minimum size=5mm]
\tikzstyle{place}=[circle, draw=black,thick,minimum size=5mm]
\node[place,tokens=1] (p1) at (0,2) {};
\node[place,tokens=1] (p3) at (3,2) {};
\node[place,tokens=1] (p5) at (6,2) {};
\node[place,tokens=1] (pc) at (3,3) {};
\node[place] (p2) at (0,0) {};
\node[place] (p4) at (3,1) {};
\node[place] (p6) at (6,0) {};
\node[transition] (t1) at (0,1) [label=left:$\pmv{b}$] {$t_1$}
edge[pre] (p1)
edge[pre, bend left] (pc)
edge[post] (p2);
\node[transition] (t2) at (1.5,1) [label=left:$\pmv{c}$] {$t_2$}
edge[pre] (p3)
edge[inhibitorblu] (p1)
edge[inhibitorblu, bend right] (p6)
edge[post] (p4);
\node[transition] (t3) at (4.5,1) [label=right:$\pmv{c}$] {$t_3$}
edge[pre] (p3)
edge[inhibitorred] (p5)
edge[inhibitorred,bend left] (p2)
edge[post] (p4);
\node[transition] (t4) at (6,1) [label=right:$\pmv{a}$] {$t_4$}
edge[pre] (p5)
edge[pre, bend right] (pc)
edge[post] (p6);
\end{tikzpicture}
 \end{center}
 The transitions labeled with $\pmv{c}$ have the same environment, namely the 
 set of labels $\setenum{\pmv{a}, \pmv{b}}$.
\end{example}
It is worth to observe that when associating a causal net to a \Ges\ we obtain a well behaved one.
\begin{proposition}\label{pr:gestoca}
  Let $\cdesname$ be a \Ges, and $\estonet{\Ges}{\ca}(\cdesname)$ be the associated contextual Petri net.
  Then $\estonet{\Ges}{\ca}(\cdesname)$ is a well behaved causal net.
\end{proposition}
\begin{proof}
 By construction of $\estonet{\Ges}{\ca}(\cdesname)$ the requirements of Definition~\ref{de:cawellbehaved}
 hold.
\end{proof}

To a causal net we can associate a triple where the relations will turn out to
be, under some further requirements, those of a \Ges.
Here the events are the labels of the transitions, conflicts between events are inferred
using the presets of the transitions and the entries are calculated using inhibitor and
read arcs. 
\begin{definition}\label{de:catoges}
 Let $\caname = \langle S, T, F, I, R, \mathsf{m}, \ell\rangle$ be a causal net labeled over 
 $\Lab = E$. 
 Define  
 $\nettoes{\ca}{\Ges}(\caname) = (E, \gesrel, \#)$ as the triple where
   \begin{itemize}
    \item $E = \ell(T)$,
    \item $\forall  e\in E$. $\pmv{Z}\gesrel e$ where 
          $\pmv{Z} = \setcomp{(X,Y)}{t\in T.\ \ell(t) =  e\ \land\ 
          X = \overrightarrow{\inib{t}}\ \land\ 
          Y = \widetilde{\rarc{t}}}$, and
    \item $\forall  e,  e'\in E$. $ e\ \#\   e$ is 
          there exists $t, t'\in T$. $\ell(t) \neq \ell(t')$ and $\pre{t}\cap\pre{t'}\neq \emptyset$.      
   \end{itemize}
\end{definition}
The construction above gives the proper \Ges, provided that the \ca\ is well behaved.
\begin{proposition}\label{pr:catoges}
  Let $\caname = \langle S, T, F, I, R, \mathsf{m}, \ell\rangle$ be a well behaved causal net and
  $\nettoes{\ca}{\Ges}(\caname) = (E, \gesrel, \#)$ the associated triple, then $\nettoes{\ca}{\Ges}(\caname)$ is
  a \Ges\ and $\Conf{\caname}{\ca} = \Conf{\nettoes{\ca}{\Ges}(\caname)}{\Ges}$.
\end{proposition}
\begin{proof}
 $\nettoes{\ca}{\Ges}(\caname)$ is clearly a \Ges. It remains to prove that 
 $\Conf{\caname}{\ca} = \Conf{\nettoes{\ca}{\Ges}(\caname)}{\Ges}$ but this follows the same argument
 of Proposition~\ref{pr:gestocaandconfig}.
\end{proof}

The following theorem assure that the notion of (well behaved) causal net is adequate in the 
relationship with context-dependent event structure.

\begin{theorem}
 Let $\caname = \langle S, T, F, I, R, \mathsf{m}, \ell\rangle$ be a well behaved causal net, and 
 let $\cdesname$ be a \Ges.
 Then $\caname \cong \estonet{\Ges}{\ca}(\nettoes{\ca}{\Ges}(\caname))$ and
 $\cdesname \equiv \nettoes{\ca}{\Ges}(\estonet{\Ges}{\ca}(\cdesname))$.
\end{theorem}
\begin{proof}
 By Propositions~\ref{pr:catoges} and~\ref{pr:gestocaandconfig}.
\end{proof}

\section{Conclusions}\label{sec:concl}
In this paper we have proposed the notion of causal net as the net counterpart of
the context-dependent event structure, and shown that the notion is adequate.
A causal net is a Petri net where the dependencies are represented
using inhibitor and read arcs and not usual flow ones: the dependencies are not any longer 
inferred by the production of a token to be consumed but they are described 
by the presence or absence of tokens in places whose meaning is that a particular 
activity has been executed or not. 
Another feature of the causal net is that an activity, identified by a label, may have several different 
implementations, but this is not surprising as usually unfoldings have this peculiarity.

Like context-dependent event structure subsumes other kinds of event structures,
also the new notion comprises other kinds of nets, and we have given a direct
translation of occurrence nets and unravel nets into causal one, and also the usual
constructions associating event structures to nets can be rewritten in this 
setting.
Like context-dependent event structures, also causal nets have a similar drawback, namely
the difficulty in understanding easily the dependencies among events, which in
some of the event structures is much more immediate.

Recently a notion of unfolding representing reversibility has been pointed out 
(\cite{MMU:coordination19}) and the issue of how find the appropriate notion
of net relating \emph{reversible} event structure has been tackled (\cite{MMPPU:rc2020})
and solved for a subclass of reversible event structure. 
The notion of causal net is then a promising one as it may be used as the basis to model 
reversible prime event structures \corr{(}{as it is done in~}\cite{MMP:lics2021}\corr{)} whereas 
classical approaches fail
(\cite{MMPPU:rc2020})\aggiunta{, as already observed}.

In this paper we have focussed on the objects and not on the relations among them, hence we have
not investigated the categorical part of the new kind of net, which we intend to
pursue in the future. A notion of morphism for \Ges\ could be 
a partial mapping $f$ on events such that if $f(e)$ is defined
and $\pmv{Z}\gesrel e$ then $f$ is such that for all $(X,Y), (X',Y')\in \pmv{Z}$ then
$f(X) = f(X') \Rightarrow\ f(Y) = f(Y')$ and the mapping can identify only conflicting events.

But when considering causal nets it should be observed that whereas read arcs are treated similarly to
flow arcs, inhibitor arcs are usually treated as contra-variant, which contrast our usage.
Thus a new different notion of mapping on nets should be introduced taking into account the
fact that causal nets are flat ones with respect to the flow relation. 

It should also be mentioned that \emph{persistent} nets have been connected to event structures
(\cite{CW:entcs05} and \cite{BBCGMM:lmcs18}), and in these nets events may happen in different
contexts, hence it would be interesting to compare these approaches to the one pursued here.

\aggiunta{
\section*{Acknowlegments}
The author would like to acknowledge that Hern\'{a}n Melgratti and Claudio A. Mezzina, 
through many fruitful discussions
on the notion of causal net, have contributed to it. 
The anonymous reviewers have to be thanked for their useful criticisms
and suggestions that have greatly contributed to improve the paper.}

\bibliographystyle{alpha}
\bibliography{main}
\end{document}